%% file: main.tex
\documentclass[11pt,letterpaper]{article}
\usepackage{amsmath, amssymb, amsthm}
\usepackage{mathtools} %
\usepackage{mathrsfs} %
\usepackage{thmtools}
\usepackage{thm-restate}
\usepackage[noadjust]{cite}
\usepackage[hidelinks,pdfencoding=auto,psdextra]{hyperref}
\hypersetup{colorlinks=false}
\hypersetup{%
	colorlinks,
	linkcolor={red!40!gray},
	citecolor={blue!40!gray},
	urlcolor={blue!70!gray}
}
\usepackage[capitalize,nameinlink]{cleveref}
\usepackage{enumitem}
\usepackage[margin=1in]{geometry}

\usepackage{algorithm}
\usepackage{algpseudocode}
\usepackage{algorithmicx}
\algrenewcommand\algorithmicrequire{\textbf{Input:}}
\algrenewcommand\algorithmicensure{\textbf{Output:}}

\newif\ifmynotes
\mynotestrue

\usepackage[dvipsnames]{xcolor}

\usepackage{graphicx}

\title{Improved Decoding of Tanner Codes\footnote{This is an extended version of the conference paper~\cite{11195639}, presented at IEEE ISIT 2025, which includes detailed proofs and an additional result, \cref{thm:viderman_intro}.}}

\author{
Zhaienhe Zhou\thanks{\texttt{zhaienhezhou@gmail.com}, School of the Gifted Young \& College of Computer Science, University of Science and Technology of China, Hefei 230026, China.}
\and
Zeyu Guo\thanks{
\texttt{zguotcs@gmail.com},
Department of Computer Science and Engineering, The Ohio State University}
}
\date{}

\newtheorem{theorem}{Theorem}[section]
\newtheorem{corollary}[theorem]{Corollary}
\newtheorem{lemma}[theorem]{Lemma}
\newtheorem{claim}[theorem]{Claim}
\newtheorem{definition}[theorem]{Definition}

\theoremstyle{remark}
\newtheorem{remark}{Remark}
\newtheorem*{remark*}{Remark}

\newcounter{casenum}

\DeclareMathOperator*{\E}{\mathbb{E}}

\newcommand{\decode}{\mathsf{InnerDecode}}
\newcommand{\erasuredecode}{\mathsf{InnerErasureDecode}}

\newcommand{\R}{\mathbb{R}}
\newcommand{\N}{\mathbb{N}}
\newcommand{\F}{\mathbb{F}}
\newcommand{\wt}{\mathrm{wt}}
\newcommand{\eps}{\varepsilon}
\renewcommand{\epsilon}{\varepsilon}

\begin{document}
\maketitle
\begin{abstract}
	In this paper, we present improved decoding algorithms for expander-based Tanner codes.

	We begin by developing a randomized linear-time decoding algorithm that, under the condition that $ \delta d_0 > 2 $, corrects up to $ \alpha n $ errors for a Tanner code $ T(G, C_0) $, where $ G $ is a $ (c, d, \alpha, \delta) $-bipartite expander with $n$ left vertices, and $ C_0 \subseteq \mathbb{F}_2^d $ is a linear inner code with minimum distance $ d_0 $. This result improves upon the previous work of Shen, Shangguan, Ouyang and Cheng (IEEE TIT 2025), which required $ \delta d_0 > 3 $.

	We further derandomize the algorithm to obtain a deterministic linear-time decoding algorithm with the same decoding radius. Our algorithm improves upon the previous deterministic algorithm of Cheng et al.\ by achieving a decoding radius of $ \alpha n $, compared with the previous radius of $ \frac{2\alpha}{d_0(1 + 0.5c\delta) }n$.

	Additionally, we investigate the size-expansion trade-off introduced by the recent work of Chen, Cheng, Li, and Ouyang (IEEE TIT 2023), and use it to provide new bounds on the minimum distance of Tanner codes. Specifically, we prove that the minimum distance of a Tanner code $T(G,C_0)$ is approximately $f_\delta^{-1} \left( \frac{1}{d_0} \right) \alpha n $, where $ f_\delta(\cdot) $ is the Size-Expansion Function.
	As another application, we improve the decoding radius of our decoding algorithms from $\alpha n$ to approximately $f_\delta^{-1}\left(\frac{2}{d_0}\right)\alpha n$.

	Finally, we extend Viderman's find-erasures-and-decode framework (ACM TOCT 2013) to general linear inner codes, obtaining a deterministic linear-time decoder for $\delta d_{0}>1.8$ when $d_0=3$, thus pushing below the $\delta d_0 > 2$ threshold of our general result.
\end{abstract}

\input{intro}

\input{prelim}

\input{rand}

\input{deter}

\input{radius}

\section{Final Remarks}

There remains a gap between the sufficient condition $\delta d_0 > 2$ required by our algorithms and the necessary condition $\delta d_0 > 1$ established in \cite{cheng2024expandercodecorrectomegan}.
We note that if $\delta d_0$ is not greater than $2$, then there exists a set $S \subseteq [n]$ such that $|N(S)| \leq \frac{2c}{d_0} |S|$. When $S$ is the set of corrupt bits, a random vertex $v\in N(S)$ is expected to have $\frac{c|S|}{|N(S)|}\geq \frac{d_0}{2}$ neighbors that are corrupt. So we cannot expect the local word $x|_{N(v)}$ to be uniquely decodable. This suggests that $\delta d_0 > 2$ might be the best achievable condition using the approach of \cite{dowlinggao18, cheng2024expandercodecorrectomegan}, where algorithms depend solely on the unique decoding of the inner code.

However, local words might still provide useful information even when unique decoding of the inner code is not feasible. For instance, the flip algorithm of Sipser and Spielman \cite{sipser1996expander} works in a setting where the inner code is a parity-check code. In this case, $d_0=2$, making unique decoding of the inner code impossible. Nevertheless, the flip algorithm proves to be effective, requiring only that $\delta \geq \frac{3}{4}$. This observation raises the question of whether insights from this special case can be generalized to achieve improved results for $d_0>2$.

As a complementary result, in \cref{sec:viderman-style} we generalize Viderman's find‑erasures‑and‑decode framework to arbitrary linear inner codes, yielding a deterministic linear‑time decoder that succeeds whenever $\delta > \frac{d_0}{2d_0-1}$.  For $d_0=3$, this requires only $\delta > 0.6$, equivalently $\delta d_0 > 1.8$, thereby improving upon the $\delta d_0>2$ threshold. Unlike the approaches of \cite{dowlinggao18,cheng2024expandercodecorrectomegan}, this framework avoids unique decoding of the inner code and instead relies on erasure decoding, which can tolerate a larger number of corrupted neighbors. It remains open whether this approach can be further developed to break the $\delta d_0>2$ barrier for $d_0>3$.

\section*{Acknowledgments}

The first author thanks Xue Chen for explaining their results \cite{chen2023improved}.
The second author thanks Chong Shangguan and Yuanting Shen for helpful discussions and for explaining their results \cite{cheng2024expandercodecorrectomegan}, and Zihan Zhang for additional discussions.

\bibliographystyle{alpha}
\bibliography{cite}

\appendix
\input{appendix}

\end{document}

%% file: intro.tex
\section{Introduction}

Low-density parity-check (LDPC) codes are an important class of error-correcting codes known for their near-optimal error-correction performance and low decoding complexity. They are characterized by sparse parity-check matrices, where most entries are zero, which makes them efficient in both encoding and decoding. LDPC codes were first introduced by Gallager \cite{ldpccode} in the 1960s but gained renewed attention due to their suitability for modern communication systems, including wireless and satellite communications, as well as for storage applications.

A key feature of LDPC codes is their representation as a bipartite graph, with bits on the left side and parity checks on the right. This leads to a parity-check graph with a constant degree, where the sparsity of the graph enables the development of efficient decoding algorithms. In particular, the belief-propagation algorithm, introduced by Gallager \cite{ldpccode}, was later proven by Zyablov and Pinsker \cite{ZyaPin75} to correct a constant fraction of errors on random LDPC codes.

Later, Tanner \cite{tannercode} extended LDPC codes by developing a recursive approach to construct long error-correcting codes from shorter inner codes. \emph{Tanner codes} are constructed by assigning a linear inner code \( C_0 \) of length \( d \) and minimum distance \( d_0 \) to the vertices of a sparse bipartite graph. Specifically, bits are placed on the left side of the bipartite graph, and each vertex on the right side is assigned an inner code that imposes constraints on the connected bits. Note that when the inner code is a parity-check code, the Tanner code simplifies to an LDPC code. This construction offers greater flexibility in code design, enabling codes with both large minimum distances and efficient decoding properties.

To analyze the decoding algorithms of LDPC and Tanner codes, Sipser and Spielman \cite{sipser1996expander} introduced the concept of vertex expansion. Expander codes are a special class of Tanner codes constructed from \((c,d,\alpha,\delta)\)-bipartite expanders, where \(c\), \(d\), \(\alpha\), and \(\delta\) are constants. Specifically, the graph \(G = (L \cup R, E)\) is left-regular of degree \(c\) and right-regular of degree \(d\), and for any set \(S \subseteq L\) with \(|S| \leq \alpha n\), the size of the neighborhood of \(S\) is at least \(\delta c |S|\). %
Expander codes are known for their efficient decoding algorithms, which can correct $\Omega(n)$ errors in linear time.
Research on expander codes in coding theory and theoretical computer science \cite{sipser1996expander, skachekexpander03,richcapacityldpc01,feldmanlpdecode05,feldmanLPdecoding07, viderman2013,vidermanlpdecoding13, chen2023improved} has focused on optimizing the decoding radius and other parameters while keeping the decoding algorithm linear-time.

Consider the special case where the inner code is a parity-check code.
In this case, the flip algorithm introduced by Sipser and Spielman \cite{sipser1996expander} can decode up to \((2\delta-1)\alpha n\) errors in linear time for any expander code with \(\delta > \frac{3}{4}\). Later, Viderman \cite{viderman2013} proposed a new decoding method, which corrects up to \(\frac{3\delta - 2}{2\delta - 1}\alpha n\) errors in linear time when \(\delta \geq \frac{2}{3}\).

More recently, Chen, Cheng, Li, and Ouyang \cite{chen2023improved} presented an improved decoding algorithm by combining previous approaches and introducing a method they term “expansion guessing.”
They also discovered a size-expansion trade-off,
which enables expansion of larger sets to be inferred from smaller sets.
Their work showed that expander codes achieve a minimum distance of \(\frac{1}{2(1-\delta)} \alpha n\), and their decoding algorithm achieves a decoding radius of \(\frac{3}{16(1-\delta)} \alpha n\), which is nearly half of the code's distance. However, their algorithm still requires \(\delta > \frac{3}{4}\) to use the flip algorithm. This raises an open question: What is the minimum \(\delta\) required to decode a linear number of errors in linear time? It was shown in \cite{viderman2013} that \(\delta > \frac{1}{2}\) is necessary.

The above studies focus on the special case of expander codes where the inner code $C_0$ is a parity-check code with minimum distance \(d_0 = 2\). Progress has also been made on the general case \cite{Chilappa10, dowlinggao18, cheng2024expandercodecorrectomegan}. Notably, Dowling and Gao \cite{dowlinggao18} proved that the condition \(d_0 \delta^2=\Omega(c)\) is sufficient for error correction using a flip-based decoding algorithm. More recently, Shen, Shangguan, Ouyang and Cheng \cite{cheng2024expandercodecorrectomegan} improved this result by showing that \(\delta d_0 > 3\) is sufficient for error correction. They also proved that \(\delta d_0 > 1\) is necessary, thereby generalizing an earlier result of Viderman \cite{viderman2013} for expander codes with parity-check inner codes.

However, many questions remain open about the optimal parameters. In particular, there is still a gap between the sufficient and necessary conditions for $\delta d_0$ to enable a linear-time decoding. In this paper, we narrow this gap by proving that \(\delta d_0 > 2\) is sufficient for expander-based Tanner codes.

\subsection{Main Results}

Let $T(G, C_0)$ be a Tanner code based on a bipartite expander $G$ and an inner code $C_0$ (see \cref{def:tanner_code}).
Our first main result is a deterministic linear-time decoding algorithm for $T(G, C_0)$.

\begin{theorem}[Informal version of \cref{thm:deter_main}]\label{thm:deter}
	Suppose \( \delta d_0 > 2 \).
	There exists a   deterministic \( O(n) \)-time algorithm that corrects up to \( \alpha n \) errors for any Tanner code \( T(G, C_0)\subseteq\F_2^n \), where \( G \) is a \( (c, d, \alpha, \delta) \)-bipartite expander and \( C_0 \) is an inner code with minimum distance \( d_0 \).
\end{theorem}

Previously, under the condition \(\delta d_0 > 3\), Shen, Shangguan, Ouyang and Cheng \cite{cheng2024expandercodecorrectomegan} gave a randomized linear-time decoding algorithm for $T(G,C_0)$ that corrects up to $\alpha n $ errors, as well as a deterministic linear-time decoding algorithm with a slightly smaller decoding radius $ \frac{2\alpha}{d_0(1 + 0.5c\delta) }n$.
\cref{thm:deter} improves on their results by (1) relaxing the condition to \(\delta d_0 > 2\), and (2) derandomizing the randomized decoding algorithm without reducing the decoding radius $\alpha n$.

To prove \cref{thm:deter}, we first establish a weaker result: Assuming \(\delta d_0 > 2\), there exists a randomized linear-time algorithm correcting $\alpha n$ errors. This result is stated as \cref{thm:rand_main}, and its proof closely follows the approach in \cite{cheng2024expandercodecorrectomegan}. Then, in \cref{sec:deter}, we derandomize the algorithm while preserving the decoding radius.

We also investigate the size-expansion trade-off introduced by \cite{chen2023improved}. Specifically, we define the Size-Expansion Function \( f_\delta(k) \) (see \cref{def:size-expansion}), which satisfies the following property: For any \( (c, d, \alpha, \delta) \)-bipartite expander, the graph remains an expander with parameters \( (c, d, k\alpha, f_\delta(k)) \), where \( k>1 \) is a constant. Consequently, our decoding algorithm achieves a decoding radius of approximately  \( f^{-1}_{\delta}\left(\frac{2}{d_0}\right) \alpha n \), which is strictly larger than \( \alpha n \).

\begin{theorem}[Informal version of \cref{thm:decode_more}]
	\cref{thm:deter} still holds with the decoding radius increased to approximately \( f^{-1}_{\delta}\left(\frac{2}{d_0}\right) \alpha n\).
\end{theorem}

Additionally, we establish the following tight bound on the minimum distance of $T(G,C_0)$:

\begin{theorem}[Informal version of \cref{thm:dis_lower,thm:dis_upper}]\label{thm:dis}
	Suppose $\delta d_0>1$.
	The minimum distance of the Tanner code \( T(G, C_0) \) is at least approximately \( f_\delta^{-1}\left(\frac{1}{d_0}\right) \alpha n \). Furthermore, this lower bound is tight in the sense that it is achieved by infinitely many examples.
\end{theorem}

The proof of \cref{thm:dis} has two parts.  First, we show that the minimum distance is bounded from below by \( f_\delta^{-1}\left(\frac{1}{d_0}\right) \alpha n \). Second, we present a construction showing that this bound is achievable when \( \alpha \) is sufficiently small but still constant. This construction builds on the approach in \cite{chen2023improved}. For details, see \cref{sec:dis}.

Recall that Viderman~\cite{viderman2013} gave a linear-time decoding algorithm for expander codes that corrects up to
$\frac{3\delta-2}{2\delta-1}\alpha n$
errors when $\delta \ge \frac{2}{3}$.
Our final result generalizes Viderman's framework from expander codes to Tanner codes with arbitrary linear inner codes, yielding the following theorem.

\begin{theorem}[Informal version of \cref{thm:viderman-style}]
	\label{thm:viderman_intro}
	Assume $d_0\ge 2$ and $\delta > \frac{d_0}{2d_0-1}$. Then there exists a deterministic $O(n)$-time decoding algorithm for $T(G,C_0)$ that corrects up to $\gamma\alpha n-O(1)$ errors, where $\gamma = \frac{\delta(2d_0-1)-d_0}{\delta d_0-1} > 0$.
\end{theorem}

In particular, for $d_0=3$, we obtain a decoder with $\delta > 3/5$, equivalently $\delta d_0 > 1.8$, improving upon the $\delta d_0 > 2$ requirement of \cref{thm:deter}. 
The proof is presented in \cref{sec:viderman-style}.

%% file: prelim.tex
\section{Preliminaries}

\subsection{Notation and Definitions}

Let $\N=\{0,1,2,\dots\}$ and $\N^+=\{1,2,\dots\}$.
For $n\in\N$, denote by $[n]$ the set $\{1,2,\dots,n\}$. The functions $\log(\cdot)$ and $\exp(\cdot)$ use base $e$.
The time complexity of algorithms is analyzed using the RAM model, in which random access of memory and arithmetic operations are assumed to take constant time. Graphs are represented using adjacency lists within the algorithms.

\paragraph{Codes.} In this paper, all codes are assumed to be \emph{Boolean linear codes}. That is, a code is defined as a subspace \(C \subseteq \F_2^n\) over the finite field \(\F_2\). The parameter \(n\) is called the \emph{length} of \(C\).

The Hamming weight of a vector \( x \in \F_2^n \), denoted \( \wt(x) \), is defined as the number of nonzero coordinates in \( x \). The Hamming distance between two vectors \( x, y \in \F_2^n \) is defined as \( d_H(x, y) \coloneq \wt(x - y) \).
The \emph{minimum distance} of a code $C$ is $d_H(C)\coloneq\min\{d_H(x,y):x,y\in C, x\neq y\}$.

\paragraph{Bipartite graphs and expanders.}

A bipartite graph \(G = (L \cup R, E)\) is called \emph{\((c,d)\)-regular} if \(\deg(u) = c\) for all \(u \in L\) and \(\deg(v) = d\) for all \(v \in R\).

For any subset of vertices \(S \subseteq L \cup R\), let \(N(S)\) denote the set of all neighbors of \(S\). Define \(N_i(S)\) as the set of vertices adjacent to exactly \(i\) vertices in \(S\). Additionally, we use the following shorthand notations for convenience:
\[
	N_{\geq i}(S) \coloneqq \bigcup_{j \geq i} N_j(S), \quad N_{\leq i}(S) \coloneqq \bigcup_{j \leq i} N_j(S).
\]

We define \( E(S, T) \) as the set of edges connecting the two vertex sets \(S\) and \(T\).

\begin{definition}[Bipartite expander]\label{def:expander_graph}
	A \((c, d, \alpha, \delta)\)-bipartite expander is a \((c, d)\)-regular bipartite graph \(G = (L \cup R, E)\) such that for any subset of vertices \(S \subseteq L\) with \(|S| \leq \alpha |L|\), it holds that \(|N(S)| \geq \delta c |S|\). 

For non-empty $S\subseteq L$, we call $\frac{N(S)}{c|S|}$ the \emph{expansion factor} of $S$. The above expansion property is equivalent to that the expansion factor of every non-empty $S\subseteq L$ of size at most $\alpha |L|$ is at least $\delta$.
\end{definition}

We now proceed to define Tanner codes, the central object studied in this paper.

\begin{definition}[Tanner code]\label{def:tanner_code}
	Let \( C_0 \) be a code of length \( d \).
	Let \( G = (L \cup R, E) \) be a \((c, d, \alpha, \delta)\)-bipartite expander, where \( L = [n] \) for some positive integer \( n \).
	For each \( v \in R \), fix a total ordering on \( N(v) \), and let \( N(v, i) \) denote its \( i \)-th element for \( i \in [d] \).
	For \( x \in \F_2^n \) and \( v \in R \), define
	\[
		x_{N(v)} \coloneqq (x_{N(v, 1)}, \dots, x_{N(v, d)}) \in \F_2^d.
	\]
	In other words, if \( x \) is viewed as an assignment \( L \to \F_2 \), then \( x_{N(v)} \) is its restriction to \( N(v) \).

	The Tanner code \( T(G, C_0) \) is a code of length \( n \), defined as
	\[
		T(G, C_0) \coloneqq \{x \in \F_2^n : x_{N(v)} \in C_0 \text{ for all } v \in R\} \subseteq \F_2^n.
	\]
	In other words, a vector \( x\in\F_2^n \) is a codeword of \( T(G, C_0) \) if for every \( v \in R \), the ``local view'' $x_{N(v)}$ is a codeword of $C_0$.
\end{definition}

Throughout this paper, we fix positive integers $c,d$ and real numbers $\alpha, \delta\in (0,1]$. The parameters $c,d,\alpha,\delta$ are viewed as constants independent of the growing parameter $n$.
Also, let \(G = (L \cup R, E)\) be a \((c, d, \alpha, \delta)\)-bipartite expander with \(L = [n]\), and let \(C_0\) be a code of length \(d\) with minimum distance \(d_0\). All lemmas and theorems are stated under the assumption that \(G\) and \(C_0\) are given, without explicitly mentioning this.

For convenience, we introduce the following definition.
\begin{definition}[Corrupt bits and unsatisfied checks]
For \(x,y \in \F_2^n\), define $F(x,y)=\{i\in [n]: x_i\neq y_i\}$. Define \(F(x)=F(x,y)\), where \(y\) is the closest codeword to \(x\) in \(T(G, C_0)\) with respect to the Hamming distance. (If there are multiple closest codewords, \(y\) is chosen to be the lexicographically smallest one.) %

Let \(U(x) \subseteq R\) denote the set of unsatisfied checks, defined as \(U(x) = \{v \in R : x_{N(v)} \not\in C_{0}\}\).
\end{definition}

\subsection{Probabilistic Tools}

We use standard tools from probability theory, such as Hoeffding's inequality and Azuma’s inequality, which are detailed in textbooks like~\cite{prob&comp}.

\begin{lemma}[Hoeffding's inequality~\cite{prob&comp}]\label{thm:chernoff_bound}
	Let \(X_1, X_2, \dots, X_n\) be independent random variables. Let \(X = \sum_{i=1}^n X_i\). If \(X_i \in [\ell, r]\) for all \(i\), then for any \(a > 0\),
	\[
		\Pr[X \geq \E[X] + a] \leq \exp{\left(-\frac{2a^2}{n(r-\ell)^2}\right)} \quad \text{and} \quad
		\Pr[X \leq \E[X] - a] \leq \exp{\left(-\frac{2a^2}{n(r-\ell)^2}\right)}.
	\]
\end{lemma}

Recall that a sequence of random variables $X_0,X_1,\dots,X_n$ is called a \emph{martingale} if for $i=0,1,\dots,n-1$,
\[
	\E[X_{i+1}|X_0,X_1,\dots,X_{i}]=X_i.
\]
Azuma’s inequality provides a concentration bound for martingales.

\begin{lemma}[Azuma's inequality~\cite{prob&comp}]\label{thm:azuma}
	Let \(X_0, X_1, \dots, X_n\) be a martingale such that \(|X_k - X_{k-1}| \leq c_k\) for all \(k\). Then, for any \(t \geq 1\) and \(\lambda > 0\),
	\[
		\Pr(|X_t - X_0| \geq \lambda) \leq 2\exp{\left(-\frac{\lambda^2}{2\sum_{k=1}^t c_k^2}\right)}.
	\]
\end{lemma}

\subsection{Auxiliary Lemmas}

We present some useful auxiliary lemmas.

\begin{lemma}\label{lem:N_let}
For any \(S\subseteq L\) with \(|S|\le \alpha n\) and integer \(t\geq 0\), it holds that
	\[
		|N_{\le t}(S)|\ge \frac{\delta(t+1)-1}{t}\cdot c|S|.
	\]
\end{lemma}

\begin{proof}
	The claim follows from a double-counting argument. Since \(G\) is left-regular of degree \(c\), we have \(c|S| =|E(S,N(S))|\). On the other hand,
	\begin{align*}
		|E(S,N(S))| & =\sum_{i=1}^{d} i|N_i(S)|                         \\
		            & \ge  |N_{\le t}(S)|+(t+1) (|N(S)|-|N_{\le t}(S)|) \\
		            & =(t+1)|N(S)|-t|N_{\le t}(S)|                      \\
		            & \geq (t+1)\delta c|S|-t|N_{\le t}(S)|,
	\end{align*}
	where the last inequality follows from the expansion property of $G$.

	Therefore, we have \(c|S|\geq (t+1)\delta c|S|-t|N_{\le t}(S)|\).
	Rearranging this establishes the lemma.
\end{proof}

\begin{lemma}\label{lem:F_U_relation}
Let $x\in \F_2^n$ and $y\in T(G,C_0)$ such that $d_H(x,y)\leq \alpha n$. 
Let $F=F(x,y)$.
Then
\[
c|F|\ge |U(x)|\ge |N_{\leq d_0-1}(F)|\ge \frac{\delta d_0 -1}{d_0-1}\cdot c|F|.
\]
\end{lemma}

\begin{proof}
Since \(G\) is left-regular of degree $c$, the number of edges incident to \(F \subseteq L\) is \(c|F|\). Each vertex in \(U(x) \subseteq R\) is incident to at least one of these edges. Therefore, \(|U(x)| \leq c|F|\).

Since the minimum distance of \(C_{0}\) is \(d_{0}\), every \(v \in R\) that is a common neighbor of no more than \(d_0 - 1\) vertices in $F$ must be an unsatisfied check, i.e., \(N_{\leq d_0-1}(F) \subseteq U(x)\). It follows that \(|N_{\leq d_0-1}(F)| \leq |U(x)|\).

Finally, applying \cref{lem:N_let} with \(t = d_0 - 1\) and \(S = F\) gives the inequality 
\[
|N_{\leq d_0-1}(F)| \geq \frac{\delta d_0 - 1}{d_0 - 1} \cdot c|F|.\qedhere
\]
\end{proof}

\begin{lemma}\label{lem:N_get}
For any $S\subseteq L$,
\[
|N_{\ge t}(S)|\leq \frac{c}{t}|S|.
\]
\end{lemma}

\begin{proof}
Since $G$ is left-regular of degree $c$, we have 
$|E(S,N(S))|=c|S|$. Moreover, by the definition of $N_{\geq t}(\cdot)$, we have
 $|E(S,N(S))|\geq t |N_{\geq t}(S)|$. The lemma follows.
\end{proof}

\begin{lemma}\label{lem:dis_lower}
    Let $T(G, C_0)$ be a Tanner code where $G$ is a $(c, d, \alpha, \delta)$-bipartite expander and the inner code $C_0$ has distance $d_0$. If $\delta d_0 >1$, then the distance of $T(G, C_0)$ is greater than $\alpha n$.
\end{lemma}
\begin{proof}
For any linear code, its distance is equal to the minimum Hamming weight among all of its codewords.

Assume, to the contrary, that $T(G, C_0)$ has a nonzero codeword $x$ such that $\wt(x) \leq \alpha n$. Let $F = \{i \in [n] : x_i \neq 0\}$, whose size is $\wt(x)$. By \cref{lem:N_let}, we have
\[
N_{\leq d_0 - 1}(F) \geq \frac{\delta d_0 - 1}{d_0 - 1} \cdot c |F| > 0.
\]
This implies that there exists $u \in R$ that has fewer than $d_0$ neighbors in $F$. Consequently, $\wt(x_{N(u)}) < d_0$. 

As the minimum distance of $C_0$ is $d_0$, we have $x_{N(u)} \notin C_0$, which contradicts the assumption that $x \in T(G, C_0)$. Thus, the lemma is proven.
\end{proof}

%% file: rand.tex
\section{Randomized Decoding}
\label{sec:rand}

In this section, we present an improved randomized flipping algorithm for the decoding regime \( \delta d_0 > 2 \), and then extend it to a full randomized decoding algorithm.

Our flipping algorithm follows the approach of \cite{cheng2024expandercodecorrectomegan}, which is based on the following idea: For each unsatisfied check \( v \in U(x) \), if \( x_{N(v)} \) is sufficiently close to a codeword \( c \in C_{0} \) of the inner code, it is likely that \( x_{N(v)} \) should be corrected to \( c \). We let each such check \( v \) cast a ``vote'' on which bits to flip.
Then, each bit is flipped with a probability determined by the votes it receives. This process corrects a constant fraction of errors. By repeating it logarithmically many times, the received word can be corrected with high probability.

Our improvement is achieved by allowing each \( v \) to send a weighted vote based on \( d_H(x_{N(v)}, y)\), where $y\in C_0$ is the closest codeword to $x_{N(v)}$, rather than using an unweighted vote when \( d_H(x_{N(v)}, y) < d_0/3 \), as was done in \cite{cheng2024expandercodecorrectomegan}.\footnote{At a high level, this bears some similarity with the GMD decoding algorithm for concatenated codes \cite{forneyGeneralizedMinimumDistance1966}, where a large $d_H(x_{N(v)},y)$ suggests that $y$ is likely incorrect.} This modification enables a tighter analysis.

\subsection{Randomized Flipping}

In the following, let \(\decode(x)\) denote the function that, given \(x \in \F_2^d\), returns the codeword \(y \in C_0\) closest to \(x\) in Hamming distance, with ties broken by selecting the lexicographically smallest \(y\).

We now present the pseudocode of the randomized flipping algorithm.

\begin{algorithm}[H]
	\caption{$\mathsf{RandFlip}(x)$}\label{alg:randflip}
	\begin{algorithmic}[1]
		\Require{$x=(x_1,\dots,x_n)\in \mathbb{F}_2^n$, where $n=|L|$.}
		\State{$t\gets  \frac{d_0}{2}$}
		\State{$(p_1,\dots,p_n)\gets (0, \cdots, 0)\in \R^n$}
		\For{each $v\in R$}
		\State{$w_v\gets \decode (x_{N(v)})$}
		\If{$1\le d_H(w_v,x_{N(v)})
			<t$}
		\State{Choose the smallest $i\in N(v)$ where $w_v$ and $x_{N(v)}$ differ}\footnotemark
		\State{$p_i\gets p_i+\frac{t-d_H (w_v,x_{N(v)})}{ct}$}
		\EndIf
		\EndFor
		\For{each $i\in [n]$}
		\State{Flip $x_i$ with probability $p_i$}
		\EndFor
		\State{\Return $x$}
	\end{algorithmic}
\end{algorithm}

\footnotetext{Implicitly, we are viewing $w_v$ and $x_{N(v)}$ as elements in $\F_2^{N(v)}: N(v)\to\F_2$ by identifying $\F_2^d$ with $\F_2^{N(v)}$ through the total ordering on $N(v)$ defined in \cref{def:tanner_code}.}

We first check that the values $p_i$ are within $[0,1]$, ensuring the validity of Line~11.

\begin{lemma}
	For \(i \in [n]\), we have \(p_i \in [0, 1]\) at Line~11.
\end{lemma}

\begin{proof}
	Each \(p_i\) is the sum of at most \(\deg(i)=c\) terms of the form \(\frac{t - d_H(w_v, x_{N(v)})}{ct}\), where each term lies between \(0\) and \(\frac{1}{c}\). The lemma follows.
\end{proof}

The following theorem states that the algorithm $\mathsf{RandFlip}(x)$ is expected to correct at least a constant fraction of errors of $x$ under certain conditions.

\begin{theorem}\label{thm:randflip_fraction}
	Assume $d_0 \delta>2$. Let $\eps_0=\frac{d_0}{2}-\frac{1}{\delta}>0$.
	Let $x\in\F_2^n$ and $y\in T(G,C_0)$ such that $d_H(x,y)\le \alpha n$.
	Let $x'$ be the output of \cref{alg:randflip} with $x$ as input.
	Then $\E[d_H(x',y)|]\leq (1-\frac{\eps_0\delta}{t}) d_H(x,y)$.
\end{theorem}

To prove \cref{thm:randflip_fraction}, we need the following lemma. Recall that $F(x,y)=\{i\in [n]: x_i\neq y_i\}$.

\begin{lemma}\label{claim:dk}
Let $x\in\F_2^n$, $y\in T(G,C_0)$, and $F=F(x,y)$. Let $v\in N_k(F)$ for some integer $k$. Let $w_v$ be as in \cref{alg:randflip}, i.e., $w_v=\decode (x_{N(v)})$.
If $w_v= y_{N(v)}$, then $d_H(w_v, x_{N(v)})=k$. On the other hand, if $w_v\neq  y_{N(v)}$, then $d_0-k\leq d_H(w_v, x_{N(v)})\leq k$. The latter case occurs only if $k\geq \frac{d_0}{2}=t$, i.e., $v\in N_{\geq t}(F)$.
\end{lemma}

\begin{proof}
By the definition of $F$ and the choice of $v$, we have $d_H(y_{N(v)}, x_{N(v)})=k$.
	As $y\in T(G, C_0)$, we have $y_{N(v)}\in C_0$.
	As $w_v$ is a vector in $C_0$ closest to $x_{N(v)}$, we have
	\[
		d_H(w_v, x_{N(v)})\leq  d_H(y_{N(v)}, x_{N(v)})=k.
	\]
	If $w_v=y_{N(v)}$, we have $d_H(w_v, x_{N(v)})=d_H(x_{N(v)}, y_{N(v)})=k$.
	On the other hand, if $w_v\neq y_{N(v)}$, then the distance between these two codewords of $C_0$ is at least $d_0$, which implies
	$d_H(w_v, x_{N(v)})\geq d_H(w_v, y_{N(v)})-d_H(x_{N(v)}, y_{N(v)})\geq d_0-k$. This proves the lemma.
\end{proof}

Now we are ready to prove \cref{thm:randflip_fraction}.

\begin{proof}[Proof of \cref{thm:randflip_fraction}]

Let $F=\{i\in [n]:x_i\neq y_i\}$, whose size is $d_H(x,y)\leq \alpha n$.
	By definition, we have
	\begin{equation}\label{eq:dh}
			d_H(x',y)  =|F|-|\{i\in F: x_i\text{ is flipped}\}|+|\{i\in [n]\setminus F: x_i \text{ is flipped}\}|.
	\end{equation}

	By \eqref{eq:dh} and linearity of expectation, we have
	\begin{equation}\label{eq:dh2}
		\E[d_H(x',y)]=|F|-\left(\sum_{i\in F} p_i-\sum_{i\in [n]\setminus F} p_i\right).
	\end{equation}

	Consider arbitrary $v\in R$. 
	In the iteration of the first \texttt{for} loop (Lines 3--9) corresponding to $v$, some $p_i$ may increase by $\frac{t-d_H (w_v,x_{N(v)})}{ct}$. We analyze how this affects the quantity $\sum_{i\in F} p_i-\sum_{i\in [n]\setminus F} p_i$. Suppose $v$ has $k$ neighbors in $F$, i.e., $v\in N_k(F)$.

	\begin{description}
		\item[Case 1:] \(k = 0\). In this case, \(d_H(w_v, x_{N(v)}) = 0\). Due to the condition \(1 \leq d_H(w_v, x_{N(v)}) < t\) at Line~5, the iteration corresponding to \(v\) does not affect \(\sum_{i \in F} p_i - \sum_{i \in [n] \setminus F} p_i\).

		\item[Case 2:] \(1 \leq k < t\). In this case, we have \(d_H(w_v, x_{N(v)}) = k \in [1, t)\) and \(w_v = y_{N(v)}\) by \cref{claim:dk}. In the iteration corresponding to \(v\), the index \(i\) chosen at Line~6 is in \(F\) since \(w_v = y_{N(v)}\). Thus, this iteration contributes exactly
		      \[
			      \frac{t - d_H(w_v, x_{N(v)})}{ct} = \frac{t - k}{ct}
		      \]
		      to \(\sum_{i \in F} p_i - \sum_{i \in [n] \setminus F} p_i\).

		\item[Case 3:] \(t \leq k< d_0\). In this case, we have \(d_H(w_v, x_{N(v)}) \geq d_0 - k = 2t - k\) by \cref{claim:dk}. Thus, the iteration corresponding to \(v\) contributes at least
		      \[
			      -\frac{t - d_H(w_v, x_{N(v)})}{ct} \geq -\frac{t - (2t - k)}{ct} = \frac{t - k}{ct}
		      \]
		      to \(\sum_{i \in F} p_i - \sum_{i \in [n] \setminus F} p_i\).

		\item[Case 4:] \(k \geq d_0\). In this case, we have \(d_H(w_v, x_{N(v)}) \geq 0\). Thus, the iteration corresponding to \(v\) contributes at least
		      \[
			      -\frac{t - d_H(w_v, x_{N(v)})}{ct} \geq -\frac{t}{ct} \geq \frac{t - k}{ct},
		      \] to \(\sum_{i \in F} p_i - \sum_{i \in [n] \setminus F} p_i\),
		      where the last inequality uses the fact that \(k \geq d_0 = 2t\).
	\end{description}
	By the above discussion, We have
	\begin{equation}\label{eq:sumofpi}
		\sum_{i\in F} p_i-\sum_{i\in [n]\setminus F} p_i\geq \sum_{k=1}^{d}\frac{t-k}{ct}|N_k(F)|.
	\end{equation}

	Next, we establish a lower bound on the RHS of \eqref{eq:sumofpi}.
	By the definition of $N_k(\cdot)$ and the fact that $G$ is left-regular of degree $c$, we have
	\begin{equation}\label{eq:eq1}
		\sum_{k=1}^{d} k|N_k(F)|=|E(F,N(F))|=c|F|.
	\end{equation}

	As $|F|\leq \alpha n$ and $G$ is a $(c,d,\alpha,\delta)$-bipartite expander, we have
	\begin{equation}\label{eq:eq2}
		\sum_{k=1}^{d} |N_k(F)|=|N(F)|\ge \delta c|F|.
	\end{equation}

	Multiplying both sides of \eqref{eq:eq2} by \(t = \frac{1}{\delta} + \epsilon_0\) and subtracting both sides of \eqref{eq:eq1}, we obtain
	\[
		\sum_{k=1}^{d}(t-k)|N_k(F)|\ge \epsilon_0 \delta c |F|,
	\]
	or equivalently,
	\begin{equation}\label{eq:eq3}
		\sum_{k=1}^{d}\frac{t-k}{ct}|N_k(F)|\ge \frac{\epsilon_0 \delta}{t}  |F|.
	\end{equation}

	Combining \eqref{eq:sumofpi} and \eqref{eq:eq3} shows
    \begin{equation}\label{eq:eq4}
        \sum_{i\in F} p_i-\sum_{i\in [n]\setminus F} p_i\geq \frac{\epsilon_0 \delta}{t} |F|.
    \end{equation}
    And \eqref{eq:dh2} and \eqref{eq:eq4} together yield $\E[d_H(x',y)]\leq (1-\frac{\eps_0\delta}{t})|F|=(1-\frac{\eps_0\delta}{t}) d_H(x,y)$, as desired.
\end{proof}

While \cref{alg:randflip} is expected to reduce the number of corrupt bits by a constant factor, the number may increase depending on the chosen randomness. Nevertheless, the following lemma shows that any such increase will not be too large. This result will be used later.

\begin{lemma}\label{lem:max_increase}
Let $x\in\F_2^n$, $y\in T(G,C_0)$, and $F=\{i\in [n]: x_i\neq y_i\}$.
The number of $i\in [n]\setminus F$ such that $p_i>0$ at the end of $\mathsf{RandFlip}(x)$ is at most $\frac{c}{t}|F|$.
In particular, for $x'=\mathsf{RandFlip}(x)$, we always have $d_H(x',y)\leq (1+\frac{c}{t})d_H(x,y)$.
\end{lemma}
\begin{proof}
Consider \(v \in U(x)\) such that the corresponding iteration increases \(p_i\) from zero to nonzero for some \(i \in [n] \setminus F\). By the way \(i\) is chosen at Line~6, we know $w_v$ and $x_{N(v)}$ differ at this bit. As \(i \in [n] \setminus F\), we know $x_{N(v)}$ and $y_{N(v)}$ agree at this bit. So \(w_v \neq y_{N(v)}\). By \cref{claim:dk}, this occurs only if \(v \in N_{\geq t}(F)\). Finally, by \cref{lem:N_get}, the number of \(v \in N_{\geq t}(F)\) is at most \(\frac{c}{t}|F|\). 
\end{proof}

\subsection{Time Complexity of Randomized Flipping}

It is straightforward to implement \cref{alg:randflip} in $O(n)$ time. However, since the main algorithm runs it $O(\log n)$ times, this is not sufficient to achieve overall linear time. Instead, we adapt \cref{alg:randflip} to run in time proportional to the number of corrupt bits, following analyses similar to those in \cite{sipser1996expander, dowlinggao18, cheng2024expandercodecorrectomegan}. 
This adjustment ensures that the total runtime across all executions of \cref{alg:randflip} forms a geometric series that sums to $O(n)$. The details are provided below for completeness.

In the following, let $x\in \F_2^n$ and $y$ be a code of $T(G, C_0)$.

\paragraph{Maintaining the set $U(x)$.}

First, we design a data structure to store the set $U(x)$ of unsatisfied checks, ensuring the following properties:
\begin{itemize}
    \item The elements in $U(x)$ can be enumerated in $O(|U(x)|)$ time.
    \item Given $v\in R$, we can determine whether $v\in U(x)$ in constant time. 
    \item Each time a bit of $x$ is flipped, $U(x)$ is updated in constant time.
\end{itemize}
Specifically, we store the elements of $U(x)$ via a linked list. Additionally, we maintain an array that records for each $v\in R$: 
\begin{enumerate}
    \item Whether $v\in U(x)$.
    \item The pointer to $v$ in the linked list.
\end{enumerate} 
The above information can be accessed in constant time for any given $v\in R$.
When flipping a bit $i\in [n]$, we update $U(x)$ in constant time by iterating over each $v\in N(i)$, updating its status in the array, and adding or removing $v$ from the linked list as needed.

\paragraph{Maintaining the set $P$.}
Define
\[
P:=\{i \in [n]: p_i\neq 0\}.
\]
By similarly using an array and a linked list, we maintain the set $P$ together with the flipping probabilities $p_i$ such that:
\begin{itemize}
    \item The elements in $P$ can be enumerated in $O(|P|)$ time.
    \item Given $i\in [n]$, we can find $p_i$ in constant time. 
    \item Each time some $p_i$ changes, the set $P$ is updated in constant time.
\end{itemize}

\paragraph{Initialization and clean-up.}

The array and the linked list used to maintain $U(x)$ are treated as global variables. They are initialized in $O(n)$ time at the start of the main algorithm, described in \cref{sec:rand-iter}, and are maintained throughout the algorithm. 

At the start of \cref{alg:randflip}, we require $p_i=0$ for all $i\in [n]$ and $P=\emptyset$. However, manually assigning these variables as in Line~3 of \cref{alg:randflip} would take $\Theta(n)$ time. Instead, we treat the array and the linked list used to maintain $p_i$ and $P$ as global variables, which are initialized in $O(n)$ time at the start of the main algorithm. Then, at the end of each execution of \cref{alg:randflip}, we reset the array and the linked list to their initial states so that they can be reused in the next execution of \cref{alg:randflip}. The resetting takes time linear in $|P|$.

\paragraph{Enumerating fewer elements.}

Beyond the previously described measures, we further modify \cref{alg:randflip} to enumerate only the unsatisfied checks and the bits with nonzero flipping probabilities.

Specifically, we replace $v\in R$ at Line~3 of \cref{alg:randflip} with $v\in U(x)$. To see that this change does not affect the flipping probabilities $p_i$, observe that for \(v \in R \setminus U(x)\), $x_{N(v)}$ is a codeword of $C_0$. Consequently, \(w_v=\mathsf{Decode}(x_{N(v)})\) is simply \(x_{N(v)}\) itself, implying that $d_H(w_v, (x_{N(v)})=0$. So $p_i$ is unchanged in the iteration corresponding to $v$.

Moreover, we modify $i\in [n]$ at Line~11 of \cref{alg:randflip} to $i\in P$, where $P=\{i\in [n]: p_i\neq 0\}$.

With these adjustments, the first loop executes at most $|U(x)|\leq c |F(x,y)|$ times, and the second loop at most $|P|\leq |U(x)|\leq c|F(x,y)|$ times. The function $\mathsf{Decode}(x_{N(v)})$ can be computed in constant time, e.g., by brute-force search. Thus, we have:

\begin{lemma}\label{lem:randflip_runtime}
 \cref{alg:randflip} can be implemented to run in $O(|F(x,y)|)$ time, where $x\in\F_2^n$ is the input and $y$ is any codeword of $T(G,C_0)$.
\end{lemma} 

\subsection{Flipping Iteratively}\label{sec:rand-iter}

We present the pseudocode of our randomized decoding algorithm:

\begin{algorithm}[H]
	\caption{$\mathsf{RandDecode}(x)$}\label{alg:randdecode}
	\begin{algorithmic}[1]
		\Require{$x=(x_1,\dots,x_n)\in \mathbb{F}_2^n$}
		\While{$|U(x)|>0$}
		\State{$x\gets\mathsf{RandFlip}(x)$}
		\EndWhile
		\State{\Return $x$}
	\end{algorithmic}
\end{algorithm}

\begin{theorem}\label{thm:rand_main}
Assume $d_0 \delta>2$ and let $\eps_0=\frac{d_0}{2}-\frac{1}{\delta}>0$.
Let $x\in\F_2^n$ and $y\in T(G, C_0)$ such that $d_H(x,y)\leq \alpha n$.
Then given the input $x$, \cref{alg:randdecode} outputs $y$ in \( O(n) \) time with probability $1-o(1)$.%
\end{theorem}

The $o(1)$ term is with respect to the growing parameter $n$. We also assume $n\geq n_0$ for some large enough constant $n_0$; otherwise, the unique decoding of $T(G, C_0)$ can be solved in constant time, e.g., via brute-force search.

\begin{proof}[Proof of \cref{thm:rand_main}]
Let $t=\frac{d_0}{2}$, $\epsilon_1=\frac{\epsilon_0 \delta}{t}$, and $\epsilon_2=\frac{c}{t}$. Let $\beta\in (0,1)$ be a small enough constant depending only on $\eps_1$ and $\eps_2$.
We consider two phases of the algorithm:

\begin{description}
\item[Phase 1: Many corrupt bits.] 

Suppose at some point during the algorithm, calling $\mathsf{RandFlip}(x)$ changes $F(x,y)$ from a set $F$ to another set $F'$, and $|F|\geq n^\beta$.%
Then 
\[
|F'|=|F|-\left(\sum_{i\in F} X_i-\sum_{i\in [n]\setminus F}X_i\right),
\]
where each $X_i$ is the indicator random variable associated with the event that $x_i$ is flipped in $\mathsf{RandFlip}(x)$. Each $X_i$ takes the value one with probability $p_i$ (see \cref{alg:randflip}) and zero with probability $1-p_i$, and these random variables $X_i$ are independent. Also note that by \cref{lem:max_increase}, the number of $i\in [n]$ such that $p_i\neq 0$ is at most $(1+\eps_2) |F|$.
Let $\mu=\E[\sum_{i\in F} X_i-\sum_{i\in [n]\setminus F}X_i]$.
By \cref{thm:randflip_fraction}, we have $\mu\geq \epsilon_1 |F|$.
Choose $k\coloneq\lceil\frac{\log(n^{1-\beta})}{-\log\left(1-\frac{\eps_1}{2}\right)}\rceil$ so that $n\cdot \left(1-\frac{\eps_1}{2}\right)^k\leq n^\beta$.
We have 
\begin{align*}
\Pr\left[|F'|\geq (1-\frac{\eps_1}{2}) |F|\right]&\leq \Pr\left[\sum_{i\in F} X_i-\sum_{i\in [n]\setminus F}X_i\leq \mu-\frac{\eps_1|F|}{2}\right]\\
&\leq \exp\left(-\frac{2\left(\frac{\eps_1}{2}|F|\right)^2}{(1+\eps_2) |F|}\right) &{(\text{Hoeffding's inequality})}\\
&\leq \exp\left(-\frac{\eps_1^2}{2(1+\eps_2)}n^\beta\right) &{(\text{since }|F|\geq n^\beta)}\\
&=o(1/k).
\end{align*}

Thus, with probability at least $1-o(1/k)$, the size of $F(x,y)$ decreases by at least a factor of $1-\frac{\eps_1}{2}$.
By the union bound, with probability $1-o(1)$, \( |F(x,y)| \) decreases by at least a factor of $1-\frac{\eps_1}{2}$ in each run of $\mathsf{RandFlip}(x)$ until $|F(x,y)|\leq n \cdot \left(1-\frac{\eps_1}{2}\right)^k\leq n^\beta$.
When this happens, the total time it takes is  $O\left(\sum_{i=0}^{k-1}\left(n\cdot \left(1-\frac{\eps_1}{2}\right)^i\right)\right)=O(n)$ by \cref{lem:randflip_runtime}.

		\item[Phase 2: Few corrupt bits.] Now assume $|F(x,y)|\leq n^\beta$. %
		Consider running $\mathsf{RandFlip}(x)$ another $\ell:=\lceil\frac{2\beta \log n}{-\log(1-\epsilon_1)}\rceil$ times. By \cref{lem:max_increase}, the size of $F(x,y)$ may increase by at most a factor of $1+\epsilon_2$ each time. Therefore, \( |F(x,y)| \) would not exceed
		      \begin{align*}
			      n^\beta \cdot (1+\epsilon_2)^\ell\le n^{0.5} <\alpha n,
		      \end{align*}
where we use the fact that $\beta$ is a small enough constant depending on $\eps_1$ and $\eps_2$.
In particular, as $|F(x,y)|\leq \alpha n$, \cref{thm:randflip_fraction} applies to each run of $\mathsf{RandFlip}(x)$, which shows that the expectation of $|F(x,y)|$ after running $\mathsf{RandFlip}(x)$ $\ell$ times is at most
\[
p\coloneq n^{\beta}\cdot (1-\epsilon_1)^\ell\leq n^{-\beta}=o(1).
\]
As the algorithm outputs $y$ whenever $|F(x,y)|=0$, by Markov's inequality, the probability that the algorithm does not output $y$ after running $\mathsf{RandFlip}(x)$ $\ell$ times is $o(1)$.
And the running time of this phase is \( O(\ell n^{0.5})=o(n) \) by \cref{lem:randflip_runtime}.

	\end{description}
Combining the two phases shows that with probability $1-o(1)$, \cref{alg:randdecode} correctly outputs $y$ in $O(n)$ time.
\end{proof}

%% file: deter.tex
\section{Deterministic Decoding}
\label{sec:deter}

In this section, we begin by applying the same derandomization technique as in \cite{cheng2024expandercodecorrectomegan} to develop a deterministic algorithm that corrects \( \gamma n \) errors when \( \delta d_0 > 2 \), where \( \gamma = \left(1+\frac{c}{t}\right)^{-1} \frac{\delta d_0 - 1}{d_0 - 1} \alpha \). Subsequently, we introduce an additional  deterministic step before this algorithm, extending the decoding radius to \( \alpha n \).

The main idea in \cite{cheng2024expandercodecorrectomegan} is as follows: The vertices in \( L = [n] \) are divided into a constant number of buckets based on their flipping probability $p_i$. It can be shown that at least one of these buckets contains a significantly higher proportion of corrupt bits than uncorrupt bits. By flipping the bits in this bucket, a small but constant fraction of the errors can be corrected.

However, the specific bucket containing the majority of corrupt bits is not known in advance. Therefore, we must recursively search through all possible choices until the number of errors is significantly reduced, as indicated by a substantial decrease in the size of \( U(x) \). We prune branches where \( |U(x)| \) does not decrease significantly. While this approach may appear to rely on brute force, careful analysis shows that the algorithm still runs in linear time.

The previous process requires the number of corrupt bits to be bounded by $\gamma n$ initially to guarantee that this number remains below \( \alpha n \) during the search, allowing the expansion property to apply. Our key new idea is that, even if the initial number of corrupt bits exceeds $\gamma n$ (but remains bounded by $\alpha n$), we can search through the first few steps to find a branch where the number of corrupt bits drops below $\gamma n$. Although we cannot immediately verify which branch works, there is only a constant number of branches. So we can run the aforementioned decoding process on all these branches and check whether any of them produces a valid codeword.

\subsection{Deterministic Flipping}

We begin by modifying \cref{alg:randflip} to obtain the following deterministic flipping algorithm.

\begin{algorithm}[H]
	\caption{$\mathsf{DeterFlip}(x,q)$}\label{alg:deterflip}
	\begin{algorithmic}[1]
    \Require{$x=(x_1,\dots,x_n)\in \mathbb{F}_2^n$ and $q\in \R$}
		\State{$t\gets  \frac{d_0}{2}$}
		\State{$p=(p_1,\dots,p_n)\gets (0, \cdots, 0)\in \R^n$}
		\For{each $v\in R$}
		\State{$w_v\gets \decode (x_{N(v)})$}
		\If{$1\le d_H(w_v,x_{N(v)})
			<t$}
		\State{Choose the smallest $i\in N(v)$ where $w_v$ and $x_{N(v)}$ differ}
		\State{$p_i\gets p_i+\frac{t-d_H (w_v,x_{N(v)})}{ct}$}
		\EndIf
		\EndFor
		\For{each $i\in [n]$}
		\State{Flip $x_i$ if $p_i=q$}
		\EndFor
		\State{\Return $x$}
	\end{algorithmic}
\end{algorithm}

\cref{alg:deterflip} is derived from \cref{alg:randflip} with the following modifications: First, it takes an additional input $q\in \mathbb{R}$ as a guess of the flipping probability. Second, instead of flipping each $x_i$ with probability $p_i$, it flips $x_i$ when $p_i$ equals $q$. In particular, \cref{alg:deterflip} is deterministic.

Define the finite set 
\[
W \coloneq \left\{\frac{i}{cd_0}: i\in\mathbb{Z}, 0\leq i\leq cd_0\right\}.
\]
Note that each $p_i\in [0,1]$ is an integral multiple of $\frac{1}{2ct}=\frac{1}{cd_0}$ and, therefore, lies within $W$.

The following lemma shows that there exists $q\in W$ such that flipping all $x_i$ with $p_i=q$ corrects a constant fraction of errors.

\begin{lemma}\label{thm:exist_m}
	Assume $d_0 \delta>2$ and let $\eps_0=\frac{d_0}{2}-\frac{1}{\delta}>0$.
	Let $x\in\F_2^n$ and $y\in T(G,C_0)$ such that $d_H(x,y)\le \alpha n$. Let $F=F(x,y)$.
    For $q\in W$, let $P_q$ be the set of $i\in [n]$ such that $p_i=q$ at the end of $\mathsf{DeterFlip}(x,q)$.
    Then there exists $q\in W\setminus\{0\}$ such that $|P_q\cap F|-|P_q\setminus F|\ge \frac{\epsilon_0 \delta }{2c t^2}|F|$.
\end{lemma}
\begin{proof}
	Assume to the contrary that \( W \) does not contain any $q$ satisfying the lemma. In other words, for every $q\in W\setminus\{0\}$, it holds that
    \begin{equation}\label{eq:lessthan}
    |P_q\cap F|-|P_q\setminus F|< \frac{\epsilon_0 \delta }{2c t^2}|F|
    \end{equation}
    Then we have
    \begin{align*}
    \sum_{i\in F} p_i - \sum_{i\in [n]\setminus F} p_i&=\sum_{q\in W}q\left(|P_q\cap F|-|P_q\setminus F|\right)\\
    &\leq \sum_{q\in W\setminus\{0\}} \left(|P_q\cap F|-|P_q\setminus F|\right)\\
    &< (|W|-1) \frac{\epsilon_0 \delta }{2c t^2}|F|\\
    &\leq \frac{\epsilon_0 \delta }{t}|F|
    \end{align*}
    where the last two inequalities hold by \eqref{eq:lessthan} and the fact that $|W|-1=cd_0=2ct>0$.
On the other hand, we know $\sum_{i\in F} p_i-\sum_{i\in [n]\setminus F} p_i\geq \frac{\epsilon_0 \delta}{t} |F|$ (see \eqref{eq:eq4} in the proof of \cref{thm:randflip_fraction}). This is a contradiction.
\end{proof}

As $\mathsf{DeterFlip}(x,q)$ only flips the bits $x_i$ with $i\in P_q$, we immediately derive the following corollary:

\begin{corollary}\label{cor:deterflip}
Under the notation and conditions in \cref{thm:exist_m}, there exists $q\in  W\setminus\{0\}$ such that $\mathsf{DeterFlip}(x,q)$ corrects at least a $\frac{\eps_0 \delta}{2ct^2}$-fraction of corrupt bits, i.e., $|F(x',y)|\leq 
\left(1-\frac{\eps_0 \delta}{2ct^2}\right) |F(x,y)|$, where $x'$ is the output of $\mathsf{DeterFlip}(x,q)$.
\end{corollary}

The proofs of \cref{lem:max_increase} and \cref{lem:randflip_runtime} still hold and yield the following counterparts.

\begin{lemma}\label{lem:max_increase2}
Let $x\in\F_2^n$ and $y\in T(G,C_0)$. For all $q\in W\setminus\{0\}$ and $x'=\mathsf{DeterFlip}(x, q)$, it holds that $d_H(x',y)\leq (1+\frac{c}{t})d_H(x,y)$, or equivalently, $|F(x',y)|\leq (1+\frac{c}{t})|F(x,y)|$.
\end{lemma}

\begin{lemma}\label{lem:deterflip_runtime}
For all $q\in W\setminus\{0\}$, \cref{alg:deterflip} can be implemented to run in $O(|F(x,y)|)$ time, where $(x,q)$ is the input and $y$ is any codeword of $T(G,C_0)$.
\end{lemma} 

\subsection{\texorpdfstring{Search for a Sequence of $q$}{Search for a Sequence of q}}

In the following, we assume $\delta d_0>2$ and let $\eps_0=\frac{d_0}{2}-\frac{1}{\delta}>0$.

The deterministic algorithm below is based on \cref{alg:deterflip}. \Cref{thm:deepflip_main} will show that it corrects a constant fraction of corrupt bits, though it is only guaranteed to work within a decoding radius somewhat smaller than $\alpha n$.

\begin{algorithm}[H]
	\caption{$\mathsf{DeepFlip}(x)$}\label{alg:deepflip}
	\begin{algorithmic}[1]
    \Require{$x=(x_1,\dots,x_n)\in \mathbb{F}_2^n$}
		\State $s \gets \left\lceil \frac{\log\left(\frac{\delta d_0 - 1}{2(d_0 - 1)}\right)}{\log(1 - \eps)} \right\rceil$, where $\eps\coloneq\frac{\eps_0 \delta}{2ct^2}$ and $t=d_0/2$.
        \State $k_{\min}\gets |R|+1$
        \State $x_{\min}\gets \perp$
		\For{each $(q_1,\dots,q_s)\in (W\setminus\{0\})^s$}
		\State{$x^{(0)}\gets x$}
		\For{$i \gets 1$ to $s$}
		\State{$x^{(i)}\gets \mathsf{DeterFlip}(x^{(i-1)},q_i)$}
		\If{$|U(x^{(i)})|>c \gamma n$}
		\State{Exit the inner loop}
        \ElsIf{$i=s$ and $|U(x^{(s)})|<k_{\min}$}
        \State{$k_{\min}\gets|U(x^{(s)})|$}
        \State{$x_{\min}\gets x^{(s)}$}
		\EndIf
		\EndFor
		\EndFor
		\State{\Return $x_{\min}$}
	\end{algorithmic}
\end{algorithm}
\begin{lemma}\label{lem:f_upper}
Let $x\in\F_2^n$ and $y\in T(G,C_0)$ such that $d_H(x,y)\le \gamma n$, where $\gamma= \left(1+\frac{c}{t}\right)^{-1} \frac{\delta d_0 - 1}{d_0 - 1} \alpha$.
Suppose $\mathsf{DeepFlip}(x)$ outputs some $x'\in\F_2^n$. Then $|F(x',y)|\leq \frac{d_0-1}{\delta d_0 -1} \gamma n$.
\end{lemma}
\begin{proof}
Suppose the output value $x_{\min}$ is assigned $x^{(s)}$ at Line~12 in the iteration corresponding to $(q_1,\dots,q_s)\in (W\setminus\{0\})^s$. We have
$x^{(0)}=x$, $x^{(i)}=\mathsf{DeterFlip}(x^{(i-1)}, q_i)$ for $i\in [s]$, and $x_{\min}=x^{(s)}$. 

As $x^{(0)}=x$, we have $|F(x^{(0)},y)|=|F(x,y)|\leq \gamma n\leq \frac{d_0-1}{\delta d_0 -1} \gamma n$.
Now, consider $i\in [s]$ and assume $|F(x^{(i-1)},y)|\leq \frac{d_0-1}{\delta d_0 -1} \gamma n$.
By \cref{lem:max_increase2},
\[
|F(x^{(i)},y)|\leq \left(1+\frac{c}{t}\right)|F(x^{(i-1)},y)|\leq \left(1+\frac{c}{t}\right)\frac{d_0-1}{\delta d_0 -1} \gamma n=\alpha n.
\]
By \cref{lem:F_U_relation}, we have
\begin{equation}\label{eq:u-and-f}
|U(x^{(i)})|\ge \frac{\delta d_0 -1}{d_0 - 1} \cdot c |F(x^{(i)},y)|.
\end{equation}
If $|F(x^{(i)},y)|>\frac{d_0-1}{\delta d_0 -1} \gamma n$, then by \eqref{eq:u-and-f}, we would have $|U(x^{(i)})|>c\gamma n$. In this case, the algorithm would exit the inner loop at Line~9 and modify $(q_1,\dots,q_s)$ to a different sequence, contradicting the choice of $(q_1,\dots,q_s)$. Therefore, $|F(x^{(i)},y)|\leq \frac{d_0-1}{\delta d_0 -1} \gamma n$.

By induction, it follows that $|F(x^{(i)},y)|\leq \frac{d_0-1}{\delta d_0 -1} \gamma n$ for $i=0,1,\dots,s$. Choosing $i=s$ proves the lemma.
\end{proof}

\begin{theorem}\label{thm:deepflip_main}
Let $x\in\F_2^n$ and $y\in T(G,C_0)$ such that $d_H(x,y)\le \gamma n$, where $\gamma= \left(1+\frac{c}{t}\right)^{-1} \frac{\delta d_0 - 1}{d_0 - 1} \alpha$.
Then $\mathsf{DeepFlip}(x)$ outputs an element $x'\in\F_2^n$ in $O(|F(x,y)|)$ time such that $|F(x',y)|\leq \frac{1}{2}|F(x,y)|$. 
\end{theorem}
\begin{proof}
	By \cref{cor:deterflip}, there exists $(q_1,\dots,q_s)\in (W\setminus\{0\})^{s}$ such that for $(x^{(0)},\dots,x^{(s)})$ defined by $x^{(0)}=x$ and $x^{(i)}=\mathsf{DeterFlip}(x^{(i-1)}, q_i)$, we have $|F(x^{(i)},y)|\leq \left(1-\frac{\eps_0 \delta}{2ct^2}\right)^i|F(x,y)|$ for $i\in [s]$. By  \cref{lem:F_U_relation}, we have $|U(x^{(i)})|\leq c|F(x^{(i)},y)|\leq c|F(x,y)| \leq c\gamma n$ for $i\in [s]$. Therefore, the iteration of the outer loop corresponding to $(q_1,\dots,q_s)$ passes the test at Line~8 for $i\in [s]$ and reaches Line~12.

    The element $x'=x_{\min}$ is chosen to minimize the number of unsatisfied checks among all branches that reach Line~12. So we have 
    \begin{equation}\label{eq:boundU}
            |U(x')|\leq |U(x^{(s)})|\leq c|F(x^{(s)},y)|\leq c\left(1-\frac{\eps_0 \delta}{2ct^2}\right)^s|F(x,y)|.
    \end{equation}
    By \cref{lem:f_upper}, we have $|F(x',y)|\leq \frac{d_0-1}{\delta d_0 -1} \gamma n\leq \alpha n$. So \cref{lem:F_U_relation} applies to $x'$ and $y$.
    Finally, 
    \[
    |F(x',y)|\leq \frac{d_0-1}{\delta d_0-1} \cdot \frac{1}{c}\cdot |U(x')|\leq \frac{d_0-1}{\delta d_0-1} \left(1-\frac{\eps_0 \delta}{2ct^2}\right)^s |F(x,y)|\leq  \frac{1}{2}|F(x,y)|
    \]
    where the first inequality holds by  \cref{lem:F_U_relation}, the second one holds by  \eqref{eq:boundU}, 
    and the last one holds by the choices of $s$.

Next, we bound the time complexity. Since $|W\setminus\{0\}|=c d_0=O(1)$ and $s=O(1)$, by \cref{lem:max_increase2} and \cref{lem:deterflip_runtime}, the running time is bounded by:
\[
O\left((cd_0)^s\cdot s\left(1+\frac{c}{d_0/2}\right)^s|F(x,y)|\right)=O(|F(x,y)|).\qedhere
\]
\end{proof}

\subsection{The Deterministic Decoding Algorithm}

We now present the deterministic decoding algorithm.
The idea is to enumerate all sequences \( (q_1,\dots, q_r)\in(W\setminus\{0\})^r \) of length $r=O(1)$. For each sequence, we iteratively apply \( x\gets \mathsf{DeterFlip}(x, q_i) \), $i=1,\dots,r$, aiming to reduce the number of corrupt bits.

There is guaranteed to be at least one sequence $(q_1,\dots,q_r)$ that reduces the number of corrupt bits to below \( \gamma n \). While we do not know which sequence achieves this, we can enumerate all possible sequences, run \( \mathsf{DeepFlip} \) repeatedly for each, and verify the final result.

In the following, let $t=d_0/2$, $\gamma= \left(1+\frac{c}{t}\right)^{-1} \frac{\delta d_0 - 1}{d_0 - 1} \alpha$, 
$\eps_0=\frac{d_0}{2}-\frac{1}{\delta}$, and
$\eps=\frac{\eps_0 \delta}{2ct^2}$.

\begin{algorithm}[H]
	\caption{$\mathsf{MainDecode}(x)$}\label{alg:maindecode}
	\begin{algorithmic}[1]
    \Require{$x=(x_1,\dots,x_n)\in \mathbb{F}_2^n$}
		\State $r \gets \left\lceil \frac{\log \gamma}{\log (1 - \epsilon)} \right\rceil$, 
		\State $r' \gets \left\lceil \log_2 (\gamma n) \right\rceil+1$
		\For{each \( (q_1,\dots,q_r) \in (W\setminus \{0\})^r \)}
		\State $\hat{x} \gets x$
		\For{$i \gets 1$ to $r$}
		\State $\hat{x} \gets \mathsf{DeterFlip}(\hat{x}, q_i)$
		\EndFor
		\For{$i \gets 1$ to $r'$}
		\State $\hat{x} \gets \mathsf{DeepFlip}(\hat{x})$
		\If{$|U(\hat{x})| > c 2^{-i} \gamma n$}
		\State $\hat{x} \gets \perp$ 
        \State Exit the inner loop
		\EndIf
		\EndFor
		\State \Return \( \hat{x} \) \textbf{if} \(\hat{x}\neq \perp\) and \( |U(\hat{x})|=0 \) and \( d_H(x, \hat{x}) \leq \alpha n \)
		\EndFor
	\end{algorithmic}
\end{algorithm}
\begin{theorem}\label{thm:deter_main}
	Assume $\delta d_0>2$. Then \cref{alg:maindecode} can be implemented to correct up to $\alpha n$ errors in $O(n)$ time for the code $T(G, C_0)$.
\end{theorem}
\begin{proof}
	We first prove the correctness of \cref{alg:maindecode}. 
    Let $x\in\F_2^n$ and $y\in T(G, C_0)$ such that $d_H(x,y)\leq \alpha n$.
    
    By \cref{cor:deterflip}, there exists a sequence \( (q_1,\dots q_r) \in (W\setminus\{0\})^r \) such that after completing the first inner loop (Lines 5--7) with respect to this sequence, we have \( |F(\hat{x}, y)|\leq \gamma n \). Fix this sequence \( (q_1,\dots q_r)\) and consider the corresponding iteration of the outer loop (assuming it is executed). By \cref{thm:deepflip_main}, \( |F(\hat{x},y)| \) is reduced by at least half each time we apply \(\hat{x}\gets \mathsf{DeepFlip}(\hat{x}) \) (Line~9).
    By \cref{lem:F_U_relation}, for $i\in [r']$, we have $|U(\hat{x})|\leq c|F(\hat{x},y)|\leq c 2^{-i} \gamma n$ at Line~10 in the $i$-th iteration of the second inner loop (Lines~8--13), which guarantees that the algorithm does not exit this inner loop at line 12. Finally, after completing the second inner loop, \( |F(\hat{x}, y)| \) will be reduced to at most $2^{-r'} \gamma n<1$,
    which implies that $F(\hat{x}, y)=0$ and hence $\hat{x}=y$. So $y$ is output at Line~15.

    This is assuming that the iteration corresponding to \( (q_1,\dots q_r)\) is executed. However, the algorithm may terminate at Line~15 earlier and output some $\hat{x}\in \F_2^n$ satisfying $|U(\hat{x})|=0$ and \( d_H(x, \hat{x}) \leq \alpha n \). The former condition means $\hat{x}\in T(G, C_0)$.
    In this case, the output is still $\hat{x}=y$ since $\alpha n$ is less than half of the minimum distance of $T(G,C_0)$ \cite{dowlinggao18}. This proves the correctness of \cref{alg:maindecode}. 

	Next, we bound the time complexity. Note that \( r,r',|W\setminus\{0\}|=O(1)\).
    By \cref{lem:deterflip_runtime}, the first inner loop runs in $O(n)$ time.
    The condition on $|U(\hat{x})|$ at Line~10 ensures that at the end of the $i$-th iteration of the second inner loop, we have $|U(\hat{x})|\leq c 2^{-i} \gamma n$ and, consequently, $|F(\hat x,y)|\leq d|U(\hat{x})|\leq dc 2^{-i} \gamma n$.
    By \cref{thm:deepflip_main}, the second inner loop also runs in $O(n+\sum_{i=1}^{r'} dc 2^{-i} \gamma n)=O(n)$ time.
    It follows that \cref{alg:maindecode} runs in $O(n)$ time.
\end{proof}

%% file: radius.tex
\section{Distance and Decoding Radius}
\label{sec:dis}

\subsection{Size-Expansion Trade-off}

In this subsection, we briefly review the size-expansion trade-off introduced in~\cite{chen2023improved}, which will be used later in this section. We also compare these results with related work and establish relevant notation.

Recall that a $(c,d,\alpha,\delta)$-bipartite expander $G=(L\cup R, E)$ satisfies the condition that for any $S \subseteq L$ with $|S| \leq \alpha  n $, we have $|N(S)| \geq \delta c |S|$.
This raises a natural question: Given this condition, can we infer a positive expansion factor $\delta'$ for a larger set $S \subseteq L$ of size $k\alpha n$ with $k>1$? %

To begin, we establish a trivial lower bound, $\delta' \geq \frac{\delta}{k}$, which was also used in the proof of \cite[Lemma~1(b)]{dowlinggao18}. Consider $S\subseteq L$ of size $k \alpha n$, and let $S' \subseteq S$ be an arbitrary subset of size $\alpha  n $.
Since every neighbor of $S'$ is also a neighbor of $S$, it follows that $|N(S)| \geq |N(S')| \geq \delta c |S'| = \delta c \alpha  n = \frac{\delta}{k} c |S|$. Thus, choosing $\delta' = \frac{\delta}{k}$ provides a valid lower bound on the expansion factor of $S$. %

A tighter bound was established in~\cite{chen2023improved} using an expectation argument. The main idea is to choose $S'$ uniformly at random among all subsets of $S$ of size $\alpha n$. Suppose, for the sake of contradiction, that $|N(S)|$ is too small. This would result in an upper bound on $\mathbb{E}[|N(S')|]$ that is too strong, contradicting the fact that every subset $S'$ of size $\alpha n$ must have an expansion factor of at least $\delta$.

We now formalize this argument.
Assume $n$ is large enough such that $d\leq k\alpha n-\alpha n$.
Consider $i\in \{0,1\dots,d\}$ and let $\beta_i \coloneqq \frac{|N_i(S)|}{c\alpha n }$ (i.e., $|N_i(S)| = \beta_i c\alpha n$). For each $v \in N_i(S)$, we calculate the probability of \( v \) being a neighbor of \( S' \):
\[
	\Pr[v \in N(S')] = 1 - \Pr[S' \cap N(v) = \emptyset]=1-\frac{\binom{k\alpha n -i}{\alpha n}}{\binom{k\alpha n}{\alpha n}} = 1 - \left(1 - \frac{1}{k}\right)^i + O\left(\frac{1}{n}\right),
\]
where the last equality
holds by \cref{lem:binom_ratio}.

Therefore, the expected number of neighbors of $S'$ satisfies:
\[
	\E[|N(S')|]  =\sum_{i=1}^d \sum_{v \in N_i(S)}\Pr[v \in N(S')]                                             =O\left(\frac{1}{n}\right)+\sum_{i=1}^{d} \left(1 - \left(1 - \frac{1}{k}\right)^i\right) \cdot \beta_i c\alpha n.
\]
On the other hand, we have $\E[|N(S')|] \geq \delta\alpha n$ by the expansion property.
So the parameters $\beta_i$ must satisfy the following constraint:
\begin{equation}\label{eq:LP_constr1}
	O\left(\frac{1}{n}\right)+\sum_{i=1}^{d} \left(1 - \left(1 - \frac{1}{k}\right)^i\right) \cdot \beta_i c\alpha n\ge \delta\alpha n.
\end{equation}

Another constraint can be derived by double counting the number of edges in $E(S,N(S))$:
\begin{align}
	c\cdot k\alpha n=|E(S,N(S))|=\sum_{i=1}^d \beta_i c\alpha n\cdot i
	\label{eq:LP_constr2}
\end{align}

Also, by definition, the number of neighbors of $S$ is:
\begin{equation}
	|N(S)|=\sum_{i=1}^d |N_i(S)|=c\alpha n \cdot \sum_{i=1}^d\beta_i
	\label{eq:LP_targ}
\end{equation}

Thus, subject to the constraints \eqref{eq:LP_constr1} and \eqref{eq:LP_constr2} in the variables $\beta_i$, the minimum of $c\alpha n \cdot \sum_{i=1}^d\beta_i$ is a lower bound on $|N(S)|$. Consequently, subject to the constraints \eqref{eq:LP_constr1} and \eqref{eq:LP_constr2}, the minimum of $\frac{1}{k} \sum_{i=1}^d\beta_i$ is a lower bound on the expansion factor $\frac{|N(S)|}{c|S|}=\frac{|N(S)|}{ck\alpha n}$ of $S$.

This can be formalized as a linear program (LP), as shown in \cite[Equations (4) and (5)]{chen2023improved}.

\begin{definition}[Degree-Dependent Size-Expansion Function \cite{chen2023improved}]\label{def:original-size-expansion}
	For $k>1$, define \(f_{d,\delta}(k)\) as the optimal value of the following LP:
	\begin{equation*}%
		\begin{aligned}
			\text{minimize} \quad   & \frac{1}{k} \sum_{i=1}^{d} \beta_i                                                       \\
			\text{subject to} \quad & \sum_{i=1}^{d} i \cdot \beta_i = k,                                                      \\
			                        & \sum_{i=1}^{d} \left(1 - \left(1 - \frac{1}{k}\right)^i\right) \cdot \beta_i \ge \delta, \\
			                        & \beta_i \ge 0, \quad \forall i\in [d].
		\end{aligned}
	\end{equation*}
	We call $f_{d,\delta}$ the Degree-Dependent Size-Expansion Function.
\end{definition}
Note that the \( O\left( \frac{1}{n} \right) \) term in the constraint \eqref{eq:LP_constr1} is omitted in this definition. However, we will see in \cref{lem:size_expansion_prop} below that this omission has a negligible effect on the LP value.

To simplify our discussion, we focus on the following LP, where the summation range is extended from $[1, d]$ to $[1, +\infty)$.

\begin{definition}[Size-Expansion Function]\label{def:size-expansion}
	For $k>1$, define \(f_{\delta}(k)\) as the optimal value of the following LP.
	\begin{equation*}%
		\begin{aligned}
			\text{minimize} \quad   & \frac{1}{k} \sum_{i=1}^{\infty} \beta_i                                                       \\
			\text{subject to} \quad & \sum_{i=1}^{\infty} i \cdot \beta_i = k,                                                      \\
			                        & \sum_{i=1}^{\infty} \left(1 - \left(1 - \frac{1}{k}\right)^i\right) \cdot \beta_i \ge \delta, \\
			                        & \beta_i \ge 0, \quad \forall i\in \mathbb{N}^+.
		\end{aligned}
	\end{equation*}
	We call $f_\delta$ the Size-Expansion Function.
\end{definition}

This modification represents a slight relaxation: The value of the new LP is less than or equal to the original, i.e., $f_{\delta}(k)\leq f_{d, \delta}(k)$, resulting in a formally weaker bound on the expansion factor. This is because any feasible solution to the original LP (with \( \beta_i = 0 \) for all \( i > d \)) remains feasible for the new LP, preserving the same objective function value.

Nevertheless, \cref{lem:size_expansion_prop} below states that $f_{d,\delta}(k) = \max \left\{ f_{\delta}(k), \frac{1}{d} \right\}$. In particular, the two are equal when $f_{\delta}(k)\geq 1/d$.
For this reason, we focus on the simpler function $f_{\delta}$, which does not depend on the parameter $d$.

\begin{restatable}{lemma}{LPprops}\label{lem:size_expansion_prop}
	The functions $f_\delta$ and $f_{d,\delta}$ satisfy the following properties:
	\begin{enumerate}%
		\item\label{p1} \( f_{\delta} \) and \( f_{d,\delta}\) are non-increasing.
		\item\label{p2} \( f_{d,\delta}(k) = \max \left\{ f_{\delta}(k), \frac{1}{d} \right\} \).
		\item\label{p3} $f_{\delta}(k)\ge\frac{\delta}{k}$.
		\item\label{p4} Omitting the \( O\left(\frac{1}{n}\right) \) term in the constraint \eqref{eq:LP_constr1} increases the value of $f_{d,\delta}(k)$ by at most \( O\left(\frac{1}{n}\right) \).
		\item\label{p5} $f_\delta(k)$ is continuous with $\lim_{k\to1}f_{\delta}(k)=\delta$ and $\lim_{k\to\infty}f_{\delta}(k)=0$.
	\end{enumerate}
\end{restatable}

\cref{lem:size_expansion_prop} is proved in \cref{sec:size_expansion_prop}. The proof also provides explicit descriptions of the functions $f_\delta(k)$ and $f_{d,\delta}(k)$ as piecewise rational functions; see \eqref{eq:f_delta_expr} and \eqref{eq:f_d_delta_expr}.

Note that $f_\delta$ and $f_{d,\delta}$ are non-increasing by \cref{lem:size_expansion_prop}. In particular, $f_{\delta}(k)$ is not only a lower bound for the expansion factor of sets of size $k\alpha n$, but also for smaller sets.
By definition, this implies the following result.

\begin{lemma}[Size-Expansion Trade-off~\cite{chen2023improved}]\label{lem:size-expansion_trade_off}
	For $k>1$, a \((c, d, \alpha, \delta)\)-bipartite expander with $n$ left vertices is also a \((c, d, k\alpha, f_\delta(k)-O(\frac{1}{n}))\)-bipartite expander.
\end{lemma}

The following result was also proved in \cite{chen2023improved}.

\begin{lemma}[\cite{chen2023improved}]
	\( f_{\delta }(k) \)  can be efficiently computed as a piecewise rational function. Specifically, $f_\delta(k) = 1 - (1 - \delta)k$ when $k \leq \frac{1}{2(1 - \delta)}$ (or equivalently, when $f_\delta(k) \geq \frac{1}{2}$).
	\label{lem:f_delta_expr_chen}
\end{lemma}

\cref{fig:size_expansion_func} shows the Size-Expansion Function $f_\delta$ for various values of $\delta$, along with the trivial lower bound $\frac{\delta}{k}$ when $\delta=0.8$. As the figure illustrates, the bound $f_\delta(k)$ is always tighter than $\frac{\delta}{k}$ for $k>1$.

\begin{figure}[H]
	\centering
	\includegraphics[width=0.9\textwidth]{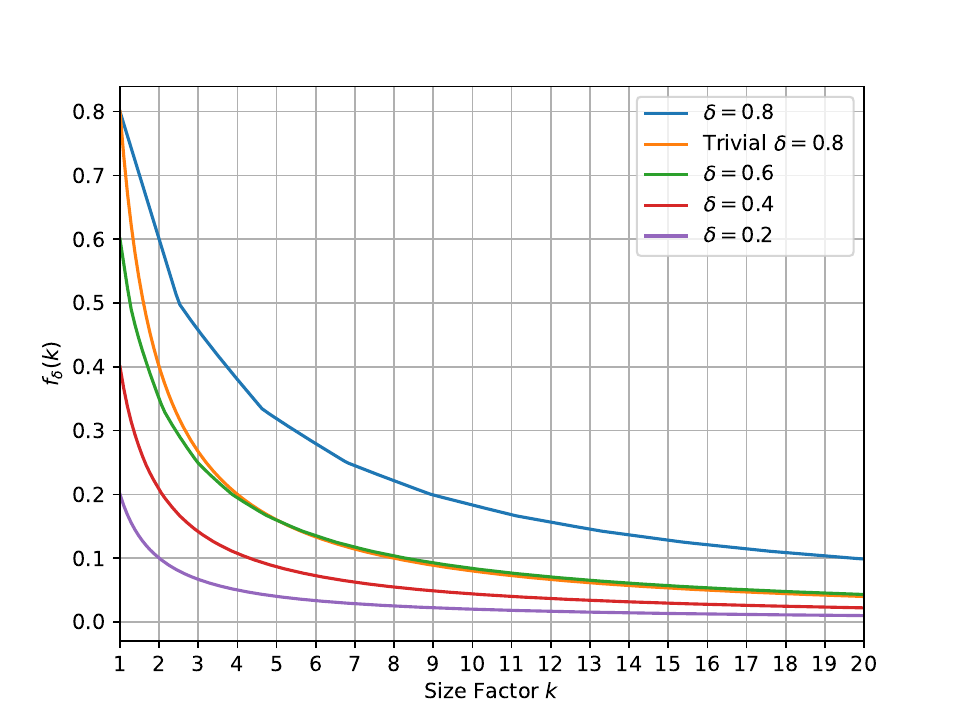}
	\caption{The Size-Expansion Function $f_\delta$. For $\delta=0.8$, the trivial bound $\frac{\delta}{k}$ is also shown.}
	\label{fig:size_expansion_func}
\end{figure}

\subsection{Bounds on Distance and Decoding Radius}

The following theorem gives a lower bound on the minimum distance of a Tanner code $T(G, C_0)$.

\begin{theorem}\label{thm:dis_lower}
	Suppose $\delta d_0>1$.
	The minimum distance of the Tanner code \( T(G,C_0) \) is greater than \(  f_\delta^{-1}\left( \frac{1}{d_0} + \epsilon \right) \alpha  n \) for any constant\footnote{It is likely that the statement holds for sub-constant $\eps$ as well. For simplicity, we do not discuss it further.} \( \epsilon \in(0,\delta-\frac{1}{d_0}) \) and all sufficiently large \( n \).
\end{theorem}

\begin{proof}
	Let $\delta' =\frac{1}{d_0} + \epsilon$, $k= f_\delta^{-1}\left(\delta'\right)$, and \(\alpha' =k \alpha\). By \cref{lem:size_expansion_prop}, $f_{\delta}^{-1}(\delta')$ is well-defined.

	Since $G$ is a $(c,d,\alpha,\delta)$-bipartite expander and  $k\alpha =\alpha', f_\delta(k)=\delta'$, applying \cref{lem:size-expansion_trade_off} implies that $G$ is also a $(c,d,\alpha', \delta')$-bipartite expander. By \cref{lem:dis_lower} and the fact that $\delta' d_0=1+d_0\epsilon>1$, we have that the distance of $T(G, C_0)$ is greater than $\alpha'$, thus proving this theorem.
\end{proof}

\begin{remark}
	\cref{thm:dis_lower} generalizes \cite[Theorem~3.1]{chen2023improved}, which establishes a lower bound of $\frac{\alpha}{2(1-\delta)} n$ on the minimum distance of expander codes, where $d_0=2$.
	To see that \cref{thm:dis_lower} recovers the this result, note that when $d_0=2$, we have
	\(f_\delta(k) = 1 - (1 - \delta)k\) when \(f_\delta(k) \geq \frac{1}{2}=\frac{1}{d_0}\) by \cref{lem:size-expansion_trade_off}.
	Therefore, $f_{\delta}^{-1}\left(\frac{1}{d_0}\right)=f_{\delta}^{-1}\left(\frac{1}{2}\right)=\frac{1}{2(1-\delta)}$.
	Substituting this into the bound gives \( f_{\delta}^{-1}\left(\frac{1}{d_0}\right)\alpha n = \frac{\alpha}{2(1-\delta)} n\).
\end{remark}

\begin{remark}
	If the Size-Expansion Function in \cref{thm:dis_lower} is replaced with the trivial bound \(\frac{\delta}{k}\), the resulting distance bound is \(\delta d_0 \alpha n\), as proved in \cite[Lemma~1(b)]{dowlinggao18}. However, this trivial bound is strictly weaker than the bound provided by \cref{thm:dis_lower}.
\end{remark}

The relationship between the distance bound and the parameters \(\delta\) and \(d_0\) is illustrated in the following figure:

\begin{figure}[H]
	\centering
	\includegraphics[width=0.9\textwidth]{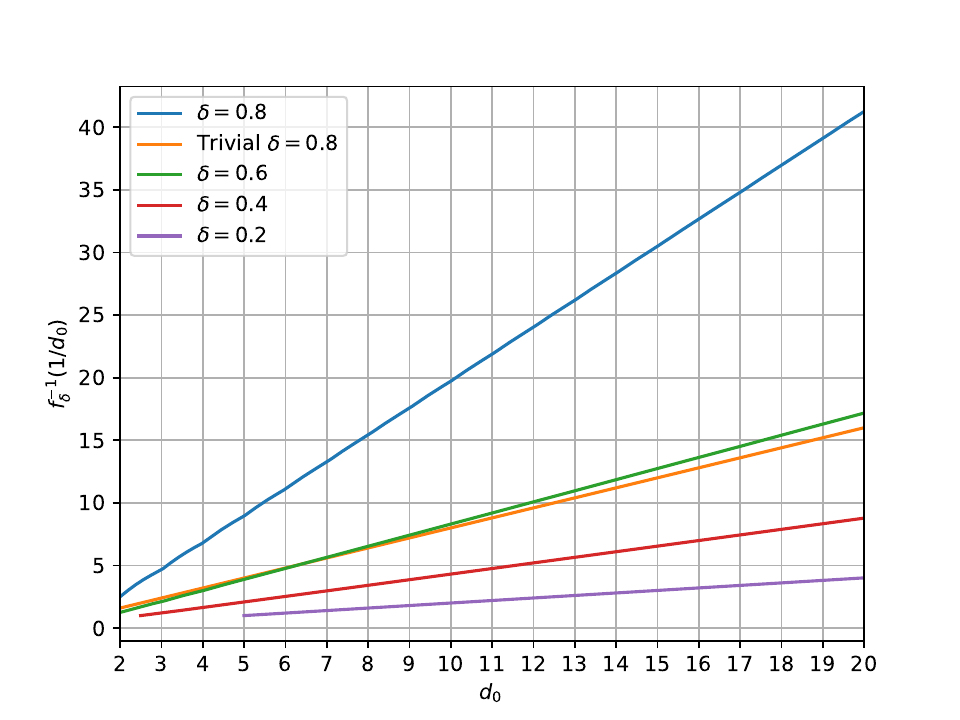}
	\caption{Plot of \(f_\delta^{-1}\left(\frac{1}{d_0}\right)\). For $\delta=0.8$, the factor $\delta d_0$ from the trivial bound $\delta d_0 \alpha n$ is also shown.}
\end{figure}

We have the following theorem on the decoding radius of our previous algorithm.
\begin{theorem}	\label{thm:decode_more}
	Suppose $\delta d_0>2$.	\cref{alg:randdecode} and \cref{alg:maindecode} can decode up to $ f_\delta^{-1}\left(\frac{2}{d_0}+\epsilon\right)\alpha n$ errors in $O(n)$ time for any constant \( \epsilon \in(0,\delta-\frac{2}{d_0}) \) and all sufficiently large \( n \).
\end{theorem}
\begin{proof}
	Let $\delta' = \frac{2}{d_0} + \epsilon$, $k= f_\delta^{-1}\left(\delta'\right)$, and \(\alpha'= k \alpha\). By \cref{lem:size_expansion_prop}, $f_{\delta}^{-1}(\delta')$ is well-defined.

	Since $G$ is a $(c,d,\alpha,\delta)$-bipartite expander and  $k\alpha =\alpha', f_\delta(k)=\delta'$, applying \cref{lem:size-expansion_trade_off} implies that the graph \(G\) is also a \((c, d, \alpha', \delta')\)-bipartite expander. Since \(\delta' d_0 = 2 + d_0 \epsilon > 2\), it follows from \cref{thm:rand_main} and \cref{thm:deter_main} that   \cref{alg:randdecode} and \cref{alg:maindecode} can both decode up to \(\alpha' n\) errors in linear time.
\end{proof}

\begin{remark}
	If the Size-Expansion Function in \cref{thm:decode_more} is replaced with the trivial bound \(\frac{\delta}{k}\), the resulting decoding radius becomes \(\frac{\delta d_0}{2} \alpha n\). This bound is strictly weaker than the one derived in \cref{thm:decode_more}.
\end{remark}

The relationship between the decoding radius and the parameters \(\delta\) and \(d_0\) is illustrated in the following figure:
\begin{figure}[H]
	\centering
	\includegraphics[width=0.9\textwidth]{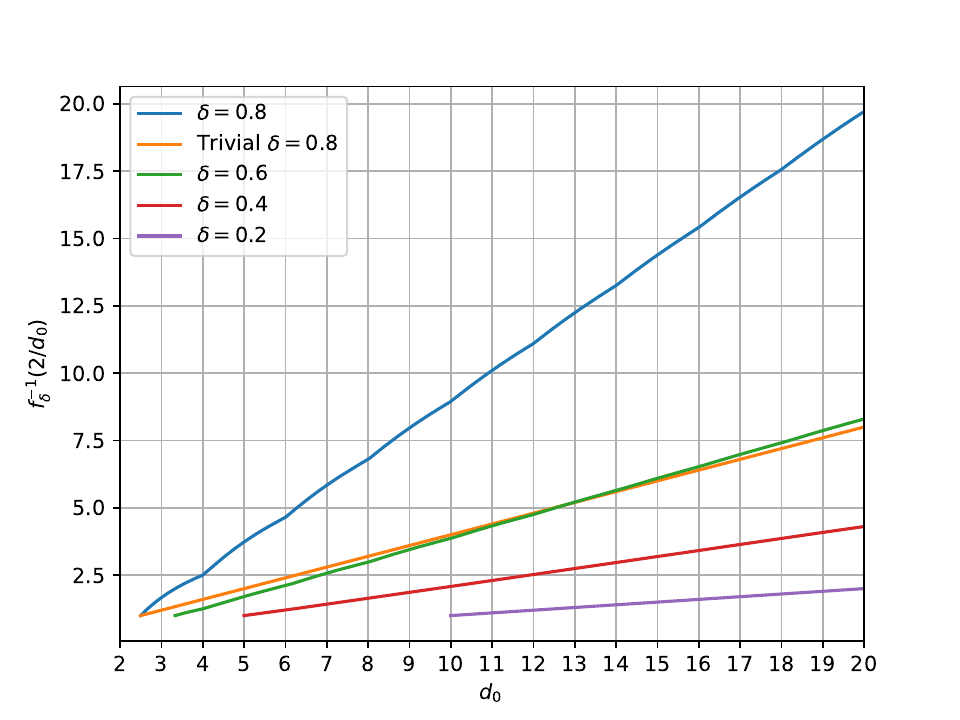}
	\caption{Plot of \(f_\delta^{-1}\left(\frac{2}{d_0}\right)\). For $\delta=0.8$, the factor $\frac{\delta d_0}{2}$ from the trivial bound $\frac{\delta d_0}{2} \alpha n$ is also shown.}
\end{figure}

\subsection{Tightness of the Distance Bound}

The following theorem, \cref{thm:dis_upper}, demonstrates that the distance bound in \cref{thm:dis} is essentially tight. This theorem follows \cite[Theorem~3.2]{chen2023improved} but extends it in two significant ways. First, while \cite[Theorem~3.2]{chen2023improved} establishes the existence of infinitely many expander codes ($d_0 = 2$) with a minimum distance of $\frac{\alpha}{2(1-\delta)} n$, the underlying graphs are only  ``almost regular'' bipartite expanders. \cref{thm:dis_upper} extends the result to strictly regular bipartite expanders. It also applies to Tanner codes $T(G, C_0)$ with $d_0 > 2$.

\begin{theorem}\label{thm:dis_upper}
	Given any constants $\delta, d_0, \epsilon > 0$ with $\delta d_0 > 1$, there exist constants $c$, $d$, and $\alpha$ such that for infinitely many values of $n$, a $(c,d,\alpha,\delta-\epsilon)$-bipartite expander $G$ with $n$ left vertices exists. Moreover, for any linear code $C_0\subseteq\F_2^d$ of minimum distance $d_0$, there exists such a graph $G$ such that, by fixing an appropriate total ordering on $N(v)$ for each $v\in R(G)$, the minimum distance of the resulting Tanner code $T(G,C_0)$ is at most $f_\delta^{-1}\left(\frac{1}{d_0}\right) \alpha n$.
\end{theorem}

\begin{proof}
	We denote the left and right parts of a bipartite graph $G$ by $L(G)$ and $R(G)$, respectively. Let $k= f_\delta^{-1}\left(\frac{1}{d_0}\right)$.

	We construct two bipartite graphs:
	\begin{itemize}
		\item $G_0$ is a uniformly chosen random $(c,d_0)$-regular graph with $|L(G_0)| = k \alpha n$ and $|R(G_0)| = c n \frac{k\alpha}{d_0}$.
		\item $G_1$ is a uniformly chosen random graph subject to the constraints that (1) $G_1$ is left-regular of degree $c$ with $|L(G_1)| = (1 - k\alpha)n$ and $|R(G_1)| = c n\frac{1}{d}$, and (2)
		      $c n \frac{k\alpha}{d_0}$ right vertices of $G_1$ have degree $d - d_0$, while the remaining ones have degree $d$. The parameter $\alpha$ is chosen to satisfy $\alpha\leq \frac{d_0}{dk}$ so that this is possible.

	\end{itemize}

	We formalize the algorithm used to generate the regular graphs $G_0$ and $G_1$ as \cref{alg:gengraph} in \cref{sec:rg}.

	Next, we merge \( G_0 \) and \( G_1 \) into a single bipartite graph \( G \). Specifically, for each vertex \( u \in R(G_0) \), we pair it with a vertex \( v\in R(G_1) \) of degree \( d - d_0 \), merging \( u \) and \( v \) into a single vertex. All other vertices in \( G_0 \) and \( G_1 \) remain unchanged. This process results in the final graph \( G \).

	\begin{figure}[H]
		\centering
		\includegraphics[width=0.9\textwidth]{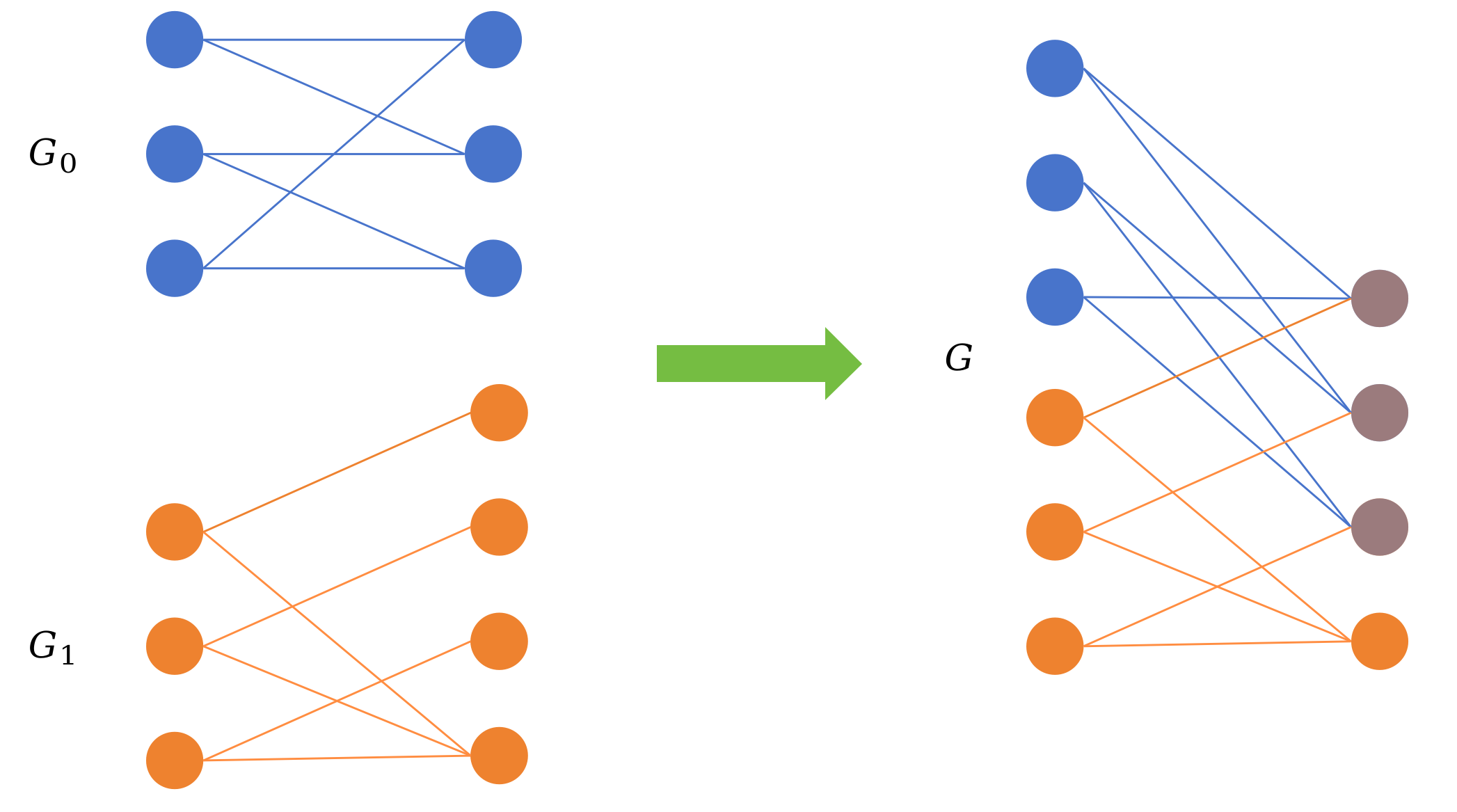}
		\caption{An illustration of merging $G_0$ and $G_1$, where $d_0=2$ and $d=3$.}
	\end{figure}

	Assume without loss of generality that $1^{d_0}0^{d-d_0}\in C_0$.
	Then, for each merged vertex $v\in R(G)$, we may order its $d$ neighbors such that the first $d_0$ neighbors are in $L(G_0)$.
	Then the unique vector in $\F_2^{L(G)}$ with support $L(G_0)$ is a codeword of $T(G,C_0)$, implying that the minimum distance of $T(G, C_0)$ is at most  $|L_0(G)|=k\alpha n$.

	It remains to prove that $G$ is a \((c,d,\alpha,\delta - \epsilon)\)-bipartite expander with high probability.
	Fix a set \( S \subseteq L(G) \) with \( |S| \leq \alpha n \). Let \( S_0 \coloneqq S \cap L(G_0) \) and \( S_1 \coloneqq S \cap L(G_1) \). Suppose \( |S_0| = \alpha_0 n \) and \( |S_1| = \alpha_1 n \), where \( \alpha_0 + \alpha_1 \eqcolon\alpha'\leq \alpha \). Next, we bound the probability that \( N(S) \) is small.
	\paragraph{Bounding the probability when $S$ is large.} Assume $\gamma n\le |S|\le \alpha n$, where $\gamma=O_{\delta,\alpha,d_{0},d}(1)$ is a small constant to be determined later. Denote the set of merged vertices in $R(G)$ as $M$.
	We calculate the expected size of $N(S)$ using Azuma's inequality.
	\begin{align}
		\E[|N(S)|] & =\notag \sum_{u \in R(G)} \Pr[u \in N(S)]                                                                                                                                                                                                                                                      \\
		         & =\notag |R(G)|-\sum_{u \in R(G), u \in M} \Pr[u \not\in N(S)] -\sum_{u \in R(G), u\not \in M} \Pr[u \not\in N(S)]                                                                                                                                                                              \\
		         & \geq c n \left(\frac{1}{d} - \frac{k\alpha}{d_0} \left(1 - \frac{\alpha_0}{k\alpha}\right)^{d_0} \left(1 - \frac{\alpha_1}{1 - k\alpha}\right)^{d-d_0} - \left(\frac{1}{d} - \frac{k\alpha}{d_0}\right)\left(1 - \frac{\alpha_1}{1 - k\alpha}\right)^{d}\right),\label{ineq:expected_neibors1}
	\end{align}
	where the last step holds since for $u\in R(G)$,
	\[
		\Pr[u \not\in N(S)]\leq\begin{cases}
			\left(1 - \frac{\alpha_0}{k\alpha}\right)^{d_0} \left(1 - \frac{\alpha_1}{1 - k\alpha}\right)^{d-d_0} & u\in M,     \\
			\left(1 - \frac{\alpha_1}{1 - k\alpha}\right)^{d}                                                     & u\not\in M.
		\end{cases}
	\]

	We claim \eqref{ineq:expected_neibors1} is bounded from below by
	\begin{align}
		 & c n\left(\frac{1}{d} - \frac{k\alpha}{d_0} \left(1 - \delta d_0 \frac{\alpha_0}{k\alpha}\right) - \left(\frac{1}{d} - \frac{k\alpha}{d_0}\right)(1 - \alpha_1 d+O_{k,d,d_0}(\alpha'^2))\right) \label{ineq:expected_neibors2} \\
		 & = c n \left(\delta \alpha_0 + \alpha_1 - O_{k,d,d_0}( \alpha'^2)\right).\label{ineq:expected_neibors2-line2}
	\end{align}
To see this, we first prove the following claim.
\begin{claim}\label{claim:new}
 \( \left( 1-\frac{1}{k} \right) ^{d_{0}}\le 1-\delta d_{0}\frac{1}{k} \).
\end{claim}
\begin{proof}
Assume to the contrary that the claim is not true. Then $ \left( 1 - \frac{1}{k} \right)^{d_0} = 1 - (1 - \eta) \delta d_0 \cdot \frac{1}{k} $ for some constant $\eta > 0$. 

Define the function $h(x)\coloneq \left( 1 - \frac{1}{k} \right)^{x}+ (1 - \eta) \delta  \frac{x}{k}-1$. Then $h(0)=0$ and $h(d_0)=0$. Note that $h(x)$ is a convex function. So if $h(t)<0$ for some $t\geq d_0$, then we would have $h(d_0)<0$, yielding a contradiction. Therefore, $h(t)\geq 0$ for all $t\geq d_0$, or equivalently,
    \begin{equation} 1 - \left( 1 - \frac{1}{k} \right)^{t} \le (1-\eta)\delta \frac{t}{k}, \quad \forall t\ge d_0. \label{eq:LP_coeff_lower_bound}
    \end{equation}

    By the proof of \cref{lem:size_expansion_prop}, the optimal solution $(\hat{\beta}_i)_{i\in\N^+}$ of $f_{\delta}(k)$ is supported on at most two adjacent indices. That is, the support is either $\{j, j+1\}$ or a singleton $\{j\}$ for some integer $j\in\N^+$.

    \begin{description}
        \item[Case 1:] $j\geq d_0$. By the constraints from \cref{def:size-expansion} and inequality \eqref{eq:LP_coeff_lower_bound}, we have
        \begin{align*}
            \delta &\le \sum_{i=1}^{\infty} \left(1 - \left(1 - \frac{1}{k}\right)^i\right) \cdot \hat{\beta}_i \\
            &= \sum_{i=d_0}^{\infty} \left(1 - \left(1 - \frac{1}{k}\right)^i\right) \cdot \hat{\beta}_i \\
            &\le (1-\eta)\frac{\delta}{k} \sum_{i=d_0}^{\infty} i \cdot \hat{\beta}_i \\
            &= (1-\eta)\frac{\delta}{k} \sum_{i=1}^{\infty} i \cdot \hat{\beta}_i \\
            &= (1-\eta)\delta,
        \end{align*}
        which yields a contradiction since $(1-\eta)\delta < \delta$.

        \item[Case 2:] $j<d_0$. By the constraints from \cref{def:size-expansion}, either $k=j\hat{\beta}_j + (j+1)\hat{\beta}_{j+1}$ or $k=j\hat{\beta}_j$, depending or whether the support of the optimal solution is $\{j, j+1\}$ or $\{j\}$.
        
        In the former case, we have
        $k<d_0 (\hat{\beta}_j+\hat{\beta}_{j+1})=k d_0 f_\delta(k)$. Equivalently,  $f_\delta(k)>\frac{1}{d_0}$. However, this contradicts the definition $k= f^{-1}_{\delta}(\frac1{d_0})$.
       
        In the latter case, we have $k<d_0 \hat{\beta}_j=k d_0 f_\delta(k)$. This similarly leads to a contradiction.
    \end{description}
Since both cases lead to a contradiction, the claim is proved.
\end{proof}

By the convexity of the function $g(x)\coloneqq (1-x)^{d_{0}}$ for $x\in [0,1]$, we know that
	\begin{equation}
		\left( 1 - \frac{\alpha_0}{k \alpha} \right)^{d_0} =g\left(\frac{\alpha_{0}}{k \alpha }\right)\le \left(1-\frac{\alpha_{0}}{\alpha }  \right) g(0)+\frac{\alpha_{0}}{\alpha }g\left(\frac{1}{k}\right)=1-\frac{\alpha_0}{\alpha}+\frac{\alpha_0}{\alpha}(1-\frac{1}{k})^{d_0}\le 1 - \delta d_0 \frac{\alpha_0}{k \alpha},
		\label{ineq:expected_neibors3}
	\end{equation}
    where the last inequality holds by \cref{claim:new}.
	By \eqref{ineq:expected_neibors3}, we see that \eqref{ineq:expected_neibors1}, and consequently $\E[|N(S)|]$, is bounded from below by  \eqref{ineq:expected_neibors2}.

	By choosing $\alpha\geq \alpha'$ to be sufficiently small (but still a constant depending on $k,d,d_0,\epsilon$), we can ensure that the term $O_{k,d,d_0}(\alpha'^2)$ in \eqref{ineq:expected_neibors2-line2} is at most $\alpha' \frac{\epsilon}{2}$. Thus, as $|S|=\alpha_0 n+\alpha_1 n=\alpha' n$, we have

	$$ \E[|N(S)|] \geq c n \left(\delta \alpha_0 + \alpha_1 - \alpha' \frac{\epsilon}{2}\right)\geq  \left(\delta -  \frac{\epsilon}{2}\right)c|S|. $$

	Now, consider the process of revealing the edges going out of $S$ one by one, in an arbitrary order. For $i=0,1,\dots, c|S|$, let
	\[
		X_i=\E [|N(S)| \mid \text{The first }i\text{ edges going out of }S\text{ are revealed}].
	\]
	Clearly, the random variables $X_i$ form a martingale, and $|X_{i+1} - X_i| \leq 1$. By Azuma's inequality (\cref{thm:azuma}), for any $\lambda>0$,
	\begin{align*}
		\Pr[|X_{c|S|} - X_0| \geq \lambda] \leq 2\exp\left(-\frac{\lambda^2}{2 c|S|}\right).
	\end{align*}

	Choose $\lambda = \sqrt{\log \left(2n^2\binom{n}{|S|}\right)\cdot2c|S|} =O\left(\sqrt{c|S|\log n+ c|S|^2\log \frac{1}{\gamma } } \right)= O_{\gamma }(\sqrt{c}|S|)$.
	Also choose sufficiently large $c=\Omega_{\gamma }(\frac{1}{\epsilon^2})$ such that $\lambda\leq \frac{\eps}{2} c|S|$. Then we have
	\begin{align}
		\Pr[|N(S)| \leq (\delta - \epsilon)c |S|]
		\leq \Pr\left[||N(S)|-\E[|N(S)|]|\geq \lambda\right]
		\leq \left(n^2 \binom{n}{|S|}\right)^{-1}.\label{ineq:prob_large_s}
	\end{align}
	\paragraph{Bounding the probability when $S$ is small.}
    Assume \(  1\le |S|< \gamma n \). We claim that with high probability, $|N(S)|\ge\delta c|S|$. By the union bound,
\begin{align*}
\Pr[|N(S)|\le \delta c|S|] & \le \sum_{\substack{T\subseteq R(G),\\|T|=\delta c|S|}}\Pr[N(S)\subseteq T]\\
& \le \binom{\frac{c}{d}n}{\delta c|S|} \cdot \max\left\{\frac{d\cdot \delta c|S|}{cn k\alpha},
		\frac{d\cdot \delta c|S|}{(1-k \alpha )n c}\right\}^{c|S|}\\
&=\binom{\frac{c}{d}n}{\delta c|S|} \cdot \max\left\{\frac{d\cdot \delta |S|}{n k\alpha},
		\frac{d\cdot \delta |S|}{(1-k \alpha )n }\right\}^{c|S|},
	\end{align*}
    where $\max\left\{\frac{d\cdot \delta c|S|}{cn k\alpha},
		\frac{d\cdot \delta c|S|}{(1-k \alpha )n c}\right\}^{c|S|}$ is an upper bound on $\Pr[N(S) \subseteq T]$. 
        
    Using the standard estimate $\binom{n}{m}\leq \left(\frac{ne}{m}\right)^m$, we obtain
    \begin{align*}
        \Pr[|N(S)|\le \delta c|S|]  & \le \left(\frac{ne}{d\delta |S|}\right)^{\delta c|S|}\cdot \max\left\{\frac{d\cdot \delta |S|}{n k\alpha},
		\frac{d\cdot \delta |S|}{(1-k \alpha )n }\right\}^{c|S|}\\
		                         & =
		\left(O_{\delta,\alpha,d_{0},d}(1)\cdot \left(\frac{|S|}{n}\right)^{1-\delta }\right)^{c|S|} 
    \end{align*}
	Choose \( \gamma  \) to be a sufficiently small constant depending on  $\delta$, $\alpha$, $d_{0}$, and $d$. And choose \( c \) to be a sufficiently large constant depending on $\delta$.\footnote{Due to the proof for large sets $S$, the parameter $c$ also needs to depend on $\gamma$ and $\eps$.} We have
	\begin{align}
		\Pr[|N(S)|\le \delta c|S|] \le \left(O_{\delta,\alpha,d_{0},d}(1)\cdot \gamma^\frac{1-\delta}{2}\left(\frac{|S|}{n}\right)^\frac{1-\delta}{2} \right)^{c|S|}  
        \leq \left(\frac{|S|}{en}\right)^{\frac{1-\delta}{2}c|S|} 
        \leq \left(\frac{|S|}{en}\right)^{2|S|}.
		\label{ineq:prob_small_s}
	\end{align}

	\paragraph{Union bound.}
	Combining the bounds \eqref{ineq:prob_large_s} and \eqref{ineq:prob_small_s} and taking the union bound over all $S$, we see that the probability that $|N(S)|\geq (\delta-\eps) c|S|$ for all $S\subseteq L(G)$ of size at most $\alpha n$ is at least
	\begin{align*}
		1 & -\Pr[\exists\,S \subseteq L(G) \text{ such that }  |S|\in [\gamma n, \alpha n] \text{ and } |N(S)|\le (\delta-\epsilon)c|S|] \\
		  & -\Pr[\exists\,S\subseteq L(G)\text{ such that } |S|\in [1,\gamma n) \text{ and } |N(S)|\le (\delta-\epsilon)c|S|]\cr
		  & \ge 1- \sum_{i=\gamma n}^n\binom{n}{i}\left(n^2\binom{n}{i}\right)^{-1}-\sum_{i=1}^{\gamma n-1}\binom{n}{i}\left(\frac{i}{en}\right)^{2i}\cr
          &\ge 1- \sum_{i=\gamma n}^n\binom{n}{i}\left(n^2\binom{n}{i}\right)^{-1}-\sum_{i=1}^{\gamma n-1}\left(\frac{ne}{i}\right)^i\left(\frac{i}{en}\right)^{2i}\cr
		  & \ge 1-o(1).
	\end{align*}
	Therefore, $G$ is a $(c,d,\alpha,\delta-\epsilon)$-bipartite expander with high probability.
\end{proof}

%% file: appendix.tex
\section{Appendix}

\subsection{Generating a Random Graph}\label{sec:rg}

The following algorithm generates a random graph given the degrees of the right vertices, while all left vertices have degree $c$.

\begin{algorithm}[H]
	\caption{$\mathsf{GenerateGraph}(\deg)$}\label{alg:gengraph}
	\begin{algorithmic}[1]
		\Require{A function $\deg: R\to\N$ }\Comment{Degrees of right vertices}
		\State $E\gets \{(v, j)\mid v \in R, j\in [\deg(v)]\}$ %
		\For{each $u \in L$}
		\For{each $i \in [c]$}
		\State Choose $(v,j) \in E$ uniformly at random
		\State Add an edge connecting $u$ and $v$
		\State $E\gets E \setminus \{(v,j)\}$
		\EndFor
		\EndFor
	\end{algorithmic}
\end{algorithm}

We claim the following property of this process and omit its proof.

\begin{lemma}
	The probability distribution of the graph generated by \cref{alg:gengraph} is independent of the order in which the pairs $ (u, i) \in L \times [c] $ are enumerated.
	\label{lem:randgen_order}
\end{lemma}

\subsection{A Combinatorial Lemma}

\begin{lemma}
	Let $ m > 0 $ and $ i\geq 0 $ be integers. Let $ k > 1 $ be a real number such that %
	$i\leq km-m$. Then,
	\[
		\left(1 - \frac{1}{k}\right)^i \geq  \frac{\binom{km-i}{m}}{\binom{km}{m}} = \left(1 - \frac{1}{k}\right)^i - O_{i,k}\left(\frac{1}{m}\right) ,
	\]
	where the $ O(\cdot) $ term is with respect to the growing parameter $ m $.
	\label{lem:binom_ratio}
\end{lemma}
\begin{proof}
	The ratio can be expressed as:
	\[
		\frac{\binom{km-i}{m}}{\binom{km}{m}} = \frac{(km-i)(km-i-1)  \ldots((k-1)m-i+1)}{km(km-1)  \ldots((k-1)m+1)}.
	\]
	Since the product $(km-i)\cdots((k-1)m+1)$ appears in both the numerator and denominator, by factoring out the $ \left(1 - \frac{1}{k}\right) $ terms, we can simplify the ratio as follows:
	$$\frac{\binom{km-i}{m}}{\binom{km}{m}} =\prod_{j=0}^{i-1}\frac{(k-1)m-j}{km-j}
		=\left( 1-\frac{1}{k} \right)^i\cdot \prod_{j=0}^{i-1}\left(1-\frac{j}{(k-1)(km-j)}\right)\leq \left(1 - \frac{1}{k}\right)^i.
	$$
	Further, we bound the additional factor:
	\begin{align*}
		\prod_{j=0}^{i-1}\left(1-\frac{j}{(k-1)(km-j)}\right) & \ge 1-\sum_{j=0}^{i-1}\frac{j}{(k-1)(km-j)} \\
		                                                      & \ge  1-\frac{i(i-1)}{2(k-1)(km-i)}
		.\end{align*}
	Therefore, we have
	\begin{align*}
		\frac{\binom{km-i}{m}}{\binom{km}{m}} & \ge  \left(1 - \frac{1}{k}\right)^i \left( 1-\frac{i(i-1)}{2(k-1)(km-i)} \right) \\
		                                      & \ge  \left(1 - \frac{1}{k}\right)^i-  \frac{i(i-1)}{2(k-1)(km-i)}                \\
		                                      & = \left(1 - \frac{1}{k}\right)^i - O_{k,i}\left(\frac{1}{m}\right).
	\end{align*}
	This completes the proof.
\end{proof}

\subsection{Properties of the Size-Expansion Function}
\label{sec:size_expansion_prop}

Next, we restate and prove \cref{lem:size_expansion_prop}.

\LPprops*

\begin{proof}
	We prove these properties separately.
	\paragraph{Monotonicity.}

	Let $ k' > k $, and let $ (\hat{\beta_i})_{i\in\N^+} $ denote the optimal solution to the LP defining $f_{\delta}(k)$. Then, $(\frac{k'}{k} \hat{\beta_i})_{i\in\N^+}$ is a feasible solution to the LP defining $f_{\delta}(k')$,
	and the corresponding value of the objective function is
	\[
		\frac{1}{k'}\sum_{i=1}^\infty \frac{k'}{k} \hat{\beta_i}=\frac{1}{k}\sum_{i=1}^\infty\hat{\beta_i}=f_{\delta}(k).
	\]
	Therefore, $f_{\delta}(k')\leq f_{\delta}(k)$. The same argument applies to $f_{d,\delta}$. This proves \cref{p1}.
	\paragraph{Relation between $f_\delta$ and $f_{d,\delta}$.}Recall the definition of $f_{\delta}(k)$ from \cref{def:size-expansion}, which is the minimum of the following LP:
	\begin{equation}
		\begin{aligned}
			\text{minimize} \quad   & \frac{1}{k} \sum_{i=1}^{\infty} \beta_i                                                       \\
			\text{subject to} \quad & \sum_{i=1}^{\infty} i \cdot \beta_i = k,                                                      \\
			                        & \sum_{i=1}^{\infty} \left(1 - \left(1 - \frac{1}{k}\right)^i\right) \cdot \beta_i \ge \delta, \\
			                        & \beta_i \ge 0, \quad \forall i.
		\end{aligned}
		\label{eq:LP}
	\end{equation}

	We now show that the optimal solution is achieved by $\hat{\beta}=(\hat{\beta}_i)_{i\in\N^+}$ that is supported on at most two adjacent indices. The proof follows \cite[Lemma~3.2]{chen2023improved}.

	The dual of this LP is:
	\begin{equation}
		\begin{aligned}
			\text{maximize} \quad   & kx + \delta y                                                                                                   \\
			\text{subject to} \quad & x \cdot i + y \left(1 - \left(1 - \frac{1}{k}\right)^i\right) \le \frac{1}{k}, \quad \forall i \in \mathbb{N}^+ \\
			                        & y \ge 0.
		\end{aligned}
		\label{eq:DP}
	\end{equation}

	Assume, for contradiction, that $\hat{\beta}_i, \hat{\beta}_j \neq 0$ with $j > i + 1$. By the complementary slackness conditions, the corresponding constraints in Equation~\eqref{eq:DP} must be satisfied as equalities:
	\[
		x \cdot i + y \left(1 - \left(1 - \frac{1}{k}\right)^i\right) = \frac{1}{k}
	\]
	and
	\[
		x \cdot j + y \left(1 - \left(1 - \frac{1}{k}\right)^j\right) = \frac{1}{k}.
	\]
	These two equalities imply that $y \neq 0$.

	The function $1 - \left(1 - \frac{1}{k}\right)^i$ is convex with respect to $i$. In particular:
	\[
		\left(1 - \left(1 - \frac{1}{k}\right)^i\right) - \left(1 - \left(1 - \frac{1}{k}\right)^{i-1}\right)
		> \left(1 - \left(1 - \frac{1}{k}\right)^{i+1}\right) - \left(1 - \left(1 - \frac{1}{k}\right)^i\right).
	\]
	From the two equalities derived from the slackness conditions, $y > 0$, and the convexity of the function, we obtain:
	\[
		x \cdot (i + 1) + y \left(1 - \left(1 - \frac{1}{k}\right)^{i+1}\right) > \frac{1}{k}.
	\]
	This contradicts the constraints of the dual program. Thus, we have shown that $\hat{\beta}$ must be supported on at most two adjacent indices.

	Now, assume that the support of $\hat{\beta} $ is contained in $\{i,i+1\}$ for some $i\in\N^+$. Then, the LP \eqref{eq:LP} simplifies to:
	\begin{align*}
		\text{minimize} \quad   & \frac{1}{k}(\beta_i + \beta_{i+1})                                                                                                                \\
		\text{subject to} \quad & \beta_i \cdot i + \beta_{i+1} \cdot (i+1) = k,                                                                                                    \\
		                        & \beta_i \cdot \left(1 - \left(1 - \frac{1}{k}\right)^i\right) + \beta_{i+1} \cdot \left(1 - \left(1 - \frac{1}{k}\right)^{i+1}\right) \ge \delta, \\
		                        & \beta_i, \beta_{i+1} \ge 0.
	\end{align*}

	The optimal value of this problem is given by:
	\begin{align}
		\begin{cases}
			\frac{\delta - (1 - \frac{1}{k})^i}{k - (k + i)(1 - \frac{1}{k})^i}, \quad & \text{for } \frac{k \left(1 - \left(1 - \frac{1}{k}\right)^{i+1}\right)}{i+1} \le \delta \le \frac{k \left(1 - \left(1 - \frac{1}{k}\right)^i\right)}{i}, \\
			\frac{1}{i+1}, \quad                                                       & \text{for } \delta < \frac{k \left(1 - \left(1 - \frac{1}{k}\right)^{i+1}\right)}{i+1},                                                                   \\
			\text{infeasible}, \quad                                                   & \text{otherwise}.
		\end{cases}
	\end{align}

	By comparing the values above for different choices of $ i $, we find the value of $f_{\delta}(k)$ is given by:
	\begin{equation}
		\label{eq:f_delta_expr}
		f_\delta(k) = \frac{\delta - (1 - \frac{1}{k})^i}{k - (k + i)(1 - \frac{1}{k})^i} \quad\text{for } i\in\N^+ \text{ such that } \frac{k \left(1 - \left(1 - \frac{1}{k}\right)^{i+1}\right)}{i+1} \le \delta \le \frac{k \left(1 - \left(1 - \frac{1}{k}\right)^i\right)}{i}.
	\end{equation}

	The same analysis applies to the LP defining $ f_{d,\delta} $, with the only difference being that the choice of $ i $ is restricted to the range $[d-1]$. Therefore, we have the expression for $ f_{d,\delta}(k) $ as follows:
	\begin{equation}
		\label{eq:f_d_delta_expr}
		f_{d,\delta}(k) =
		\begin{cases}
			\frac{1}{d}, \quad                                                         & \text{for } \delta < \frac{k \left(1 - \left(1 - \frac{1}{k}\right)^d\right)}{d},                                                                                                       \\
			\frac{\delta - (1 - \frac{1}{k})^i}{k - (k + i)(1 - \frac{1}{k})^i}, \quad & \text{for } i\in [d-1] \text{ such that } \frac{k \left(1 - \left(1 - \frac{1}{k}\right)^{i+1}\right)}{i+1} \le \delta \le \frac{k \left(1 - \left(1 - \frac{1}{k}\right)^i\right)}{i}.
		\end{cases}
	\end{equation}

	Next, we consider the condition $ f_\delta(k) \neq f_{d,\delta}(k) $, which implies that the support of the optimal solution $\hat{\beta}$ of the LP \eqref{eq:LP} has the form $\{i+1\}$ or $\{i,i+1\}$ with $i\geq d$. Therefore, comparing \eqref{eq:f_d_delta_expr} with \eqref{eq:f_delta_expr} proves \cref{p2}.

	\paragraph{Improvement over the trivial bound $\frac{\delta}{k} $.}
	By substituting $ x = 0 $ and $ y = \frac{1}{k} $ into the dual program \eqref{eq:DP}, we see that $ f_{\delta}(k)\geq \frac{\delta}{k} $, which establishes \cref{p3}.

	\paragraph{Effect of omitting the $ O\left(\frac{1}{n}\right) $ term.}
	Let $ f^*_{d,\delta}(k) $ be the size-expansion trade-off defined by the LP without omitting the error term. The only difference from the LP defining $ f_{d,\delta}(k) $ is that the constraint becomes $ \sum_{i=1}^{d} \left(1 - \left(1 - \frac{1}{k}\right)^i\right) \beta_i \ge \delta - O\left(\frac{1}{n}\right) $, and the dual objective function changes to $ kx + \left(\delta - O\left(\frac{1}{n}\right)\right)y $.

	Let $ (\hat{x}, \hat{y}) $ be the optimal solution for the dual program for $ f_{d,\delta}(k) $, so $ f_{d,\delta}(k) = k\hat{x} + \delta \hat{y} $. Since $ \hat{x}, \hat{y} $ also satisfy the dual program for $ f^*_{d,\delta}(k) $, we obtain the bound:
	\[
		f^*_{d,\delta}(k) \ge f_{d,\delta}(k) - \hat{y}\cdot  O\left( \frac{1}{n} \right).
	\]

	Next, we bound $\hat y$. The constraint in \eqref{eq:DP} with $ i = 1 $ states that $ \hat{x} + \frac{\hat{y}}{k} \le \frac{1}{k} $, i.e., $ \hat{x} \le \frac{1 - \hat{y}}{k} $. Therefore,
	\[
		f_{d,\delta}(k) = k \hat{x} + \delta \hat{y} \le 1 - \hat{y} + \delta \hat{y}.
	\]
	By \cref{p3}, $f_{d,\delta}(k)\geq \frac{\delta}{k}\geq 0$. So $y\leq \frac{1}{1-\delta}=O(1)$.
	Thus, $f^*_{d,\delta}(k) \ge f_{d,\delta}(k) -   O\left( \frac{1}{n} \right)$, which
	proves \cref{p4}.

	\paragraph{Continuity and range.} The continuity of $f_\delta$ can be established via the explicit description of $f_\delta$ given in \eqref{eq:f_delta_expr}.
	When $k$ is located at the endpoint of two pieces, namely, when
	\[
		\frac{k \left(1 - \left(1 - \frac{1}{k}\right)^{i+1}\right)}{i+1} = \delta
	\]
	for some positive integer $i$, we have
	\[
		\lim_{x \to k^-} f_{\delta}(x) = \frac{\delta - \left(1 - \frac{1}{k}\right)^i}{k - (k + i)\left(1 - \frac{1}{k}\right)^i} = \frac{1}{i+1} = \frac{\delta - \left(1 - \frac{1}{k}\right)^{i+1}}{k - (k + i+1)\left(1 - \frac{1}{k}\right)^{i+1}} = \lim_{x \to k^+} f_{\delta}(x).
	\]
	Thus, $f_\delta$ is continuous.

	To prove that  $\lim_{k \to 1} f_{\delta}(k) = \delta$, note that for bounded $k$, the value of $i$ chosen in \eqref{eq:f_delta_expr}
	satisfies $i=O(1/\delta)$.
	Thus,
	\[
		\lim_{k \to 1} f_{\delta}(k)=\lim_{k\to 1} \frac{\delta - (1 - \frac{1}{k})^i}{k - (k + i)(1 - \frac{1}{k})^i} =\lim_{k\to 1}\frac{\delta}{k}=\delta.
	\]

	Next, we prove that $\lim_{k \to \infty} f_{\delta}(k) = 0$. For any fixed $i$, we have
	\[
		\lim_{k \to \infty} \frac{k \left(1 - \left(1 - \frac{1}{k}\right)^{i+1}\right)}{i+1} = 1.
	\]
	Thus, as $k \to \infty$, the value of $i$ chosen in \eqref{eq:f_delta_expr} also tends to infinity. Therefore, the value of the function satisfies
	\[
		\lim_{k \to \infty} f_\delta(k) = \lim_{k \to \infty} \frac{\delta - \left(1 - \frac{1}{k}\right)^i}{k - (k + i)\left(1 - \frac{1}{k}\right)^i} \le \lim_{k \to \infty} \frac{1}{i} = 0,
	\]
	which completes the proof of \cref{p5}.
\end{proof}

\section{Generalizing Viderman's Algorithm to Tanner Codes}
\label{sec:viderman-style}

In this section, we generalize Viderman's find-erasures-and-decode algorithm~\cite{viderman2013}. Viderman's algorithm first identifies a small superset of the corrupt bits and then decodes this set as erasures. We show that this approach extends to arbitrary linear inner codes, yielding a deterministic linear-time decoder in the regime $\delta>\frac{d_0}{2d_0-1}$. In particular, when $d_0=3$, this yields the below-two regime $\delta d_0>1.8$.

Throughout this section, we set the threshold parameter:
\begin{equation}
	\label{eq:viderman-h}
	h\coloneq \frac{c(\delta d_0-1)}{d_0-1}.
\end{equation}
Note that $h\le c$, and under the condition $\delta>\frac{d_0}{2d_0-1}$ we have $h>c(1-\delta)$. The latter inequality ensures that the find-erasure step identifies a superset of the corrupt set $F$ that is sufficiently close to $F$.

We first present the algorithm for finding erasures.

\begin{algorithm}[H]
	\caption{$\mathsf{FindErasure}(x)$}\label{alg:finderasure}
	\begin{algorithmic}[1]
		\Require{$x=(x_1,\dots,x_n)\in \F_2^n$}
		\State $h\gets \frac{c(\delta d_0-1)}{d_0-1}$
		\State $L'\gets \emptyset$, $R'\gets U(x)$
		\While{$\exists\ i\in L\setminus L' \text{ such that } |N(i)\cap R'|\geq h$}
		\State $L'\gets L' \cup\{i\}$
		\State $R'\gets R' \cup N(i)$
		\EndWhile
		\State \Return $L'$
	\end{algorithmic}
\end{algorithm}

Next, we describe the erasure decoding algorithm $\mathsf{ErasureDecode}(x,L')$, which takes as input the received word $x$ and a set $L'\subseteq L$ that is intended to be a superset of $F$.

This algorithm invokes a subroutine $\erasuredecode$, the erasure decoder for the inner code $C_0$. Formally, given a local word $z\in\F_2^d$ and a set of erased coordinates $E\subseteq[d]$ with $|E|<d_0$, $\erasuredecode(z,E)$ returns the unique codeword $\omega\in C_0$, if it exists, such that $\omega_i=z_i$ for all $i\notin E$.

In the algorithm below, the set $N(v)\cap L'$ is identified with the corresponding local coordinate positions of $C_0$ when evaluating $\erasuredecode(\hat{x}_{N(v)},N(v)\cap L')$.

\begin{algorithm}[H]
	\caption{$\mathsf{ErasureDecode}(x,L')$}\label{alg:erasuredecode}
	\begin{algorithmic}[1]
		\Require{$x=(x_1,\dots,x_n)\in \F_2^n$ and $L'\subseteq L$}
		\State $\hat x\gets x$
		\While{$L'\neq \emptyset$}
		\State Find $v\in R$ such that $1\leq |N(v)\cap L'|\leq d_0-1$
		\State $\hat x_{N(v)}\gets \erasuredecode(\hat x_{N(v)},N(v)\cap L')$
		\State $L'\gets L'\setminus N(v)$
		\EndWhile
		\State \Return $\hat x$
	\end{algorithmic}
\end{algorithm}

The main algorithm first applies \cref{alg:finderasure} to identify a superset of $F$, and then applies \cref{alg:erasuredecode} to decode the erasures.

\begin{algorithm}[H]
	\caption{$\mathsf{MainDecode}(x)$}\label{alg:vider_maindecode}
	\begin{algorithmic}[1]
		\Require{$x=(x_1,\dots,x_n)\in \F_2^n$}
		\State \Return $\mathsf{ErasureDecode}(x, \mathsf{FindErasure}(x))$
	\end{algorithmic}
\end{algorithm}

The main result of this section is the following theorem:
\begin{theorem}\label{thm:viderman-style}
	Assume $d_0\ge 2$ and $\delta>\frac{d_0}{2d_0-1}$. Then \cref{alg:vider_maindecode} can be implemented in $O(n)$ time to correct up to $\gamma(\alpha n-1)$ errors, where
	\[
		\gamma
		\coloneq 1-\frac{c(1-\delta)}{h}
		=\frac{\delta(2d_0-1)-d_0}{\delta d_0-1}>0.
	\]
\end{theorem}

Before proving \cref{thm:viderman-style}, we present two lemmas analyzing the properties of the set $L'$ output by \cref{alg:finderasure}.
\begin{lemma}\label{lem:finderasure-small}
	Let $y\in T(G,C_0)$, let $F=F(x,y)$, and let $L'$ be the output of \cref{alg:finderasure}. Assume $|F|\le \gamma(\alpha n-1)$, where $\gamma$ is as in \cref{thm:viderman-style}. Then $|L'|\le \alpha n$.
\end{lemma}

\begin{proof}
In the proof, we use $L'$ and $R'$ to denote the intermediate sets at any point of executing the algorithm. We prove that the invariant $|L'|\le \alpha n -1$ is maintained throughout the whole procedure.

If, during the execution of $\mathsf{FindErasure}(x)$, there are multiple valid choices for the next vertex to be added to $L'$, the order in which they are added does not affect the final set $L'$. Indeed, since $R'$ only grows throughout the algorithm, any vertex $i$ satisfying $|N(i)\cap R'|\ge h$ at some point will continue to satisfy this condition until it is added to $L'$.

Let
\[
    F_0\coloneq \{i\in F: |N(i)\cap U(x)|\ge h\}.
\]
We choose an order in which all vertices in $F_0$ are added before any other vertices. This is valid since every vertex in $F_0$ satisfies the threshold condition at the beginning of the algorithm. 

Since $|F|\le \gamma(\alpha n-1)$ and  $\gamma\le 1$, we have
\[ |F_0|\le |F|\le \alpha n-1. \]
So the invariant holds for $L'=F_0$.

Every check node in $U(x)$ has at least one neighbor in $F$. Therefore, $|E(F,U(x))|\ge |U(x)|$. On the other hand, every vertex in $F\setminus F_0$ contributes fewer than $h$ edges to $E(F, U(x))$. Hence
\[
    |E(F_0,U(x))|\geq |E(F,U(x))|-h(|F|-|F_0|) \geq|U(x)|-h(|F|-|F_0|).
\]

Therefore, the size of $R'$ after adding all vertices in $F_0$, denoted by $R_0$, is bounded by
\begin{equation}
    \begin{aligned}
        |R_0|
         & = |U(x)\cup N(F_0)|                      \\
         & \le |U(x)|+|N(F_0)\setminus U(x)|        \\
         & \le |U(x)|+|E(F_0,N(F_0)\setminus U(x))| \\
         & \le |U(x)|+c|F_0|-|E(F_0,U(x))|          \\
         & \le (c-h)|F_0|+h|F|.
        \label{eq:R0-upper}
    \end{aligned}
\end{equation}

After all the vertices of $F_0$ have been added, we consider the remaining vertices added to $L'$ one by one. Whenever a vertex $i$ is added to $L'$, it already has at least $h$ neighbors in $R'$, and hence at most $c-h$ new right vertices are added to $R'$. So we always have
\begin{equation}
    \label{eq:viderman-R-upper}
    |R'|\leq |R_0|+(c-h)(|L'|-|F_0|)\leq(c-h)|L'|+h|F|,
\end{equation}
where the second inequality follows from \eqref{eq:R0-upper}. By the invariant, $|L'|\le \alpha n-1$ before adding a single vertex. So after that, we still have $|L'|\leq \alpha n$, and the expansion property gives
\[
    \delta c|L'|\leq |N(L')|\leq |R'|\le(c-h)|L'|+h|F|,
\]
where the last inequality holds by \eqref{eq:viderman-R-upper}.
Therefore,
\[
    \left(h-c(1-\delta)\right)|L'|\le h|F|.
\]
As $|F|\le \gamma(\alpha n-1)$ and $\gamma=1-\frac{c(1-\delta)}{h}$, we have
\[
    |L'|\le \frac{h|F|}{h-c(1-\delta)}
    \leq \frac{h\gamma(\alpha n-1)}{h-c(1-\delta)}
    =\alpha n -1.
\]
Repeating this argument for each vertex added to $L'$ shows that the invariant $|L'|\leq \alpha n-1$ holds throughout the algorithm, and in particular for the final output.
\end{proof}

\begin{lemma}\label{lem:finderasure-contain}
	Let $y\in T(G,C_0)$, let $F=F(x,y)$, and let $L'$ be the output of \cref{alg:finderasure}. If $|F|\leq \alpha n$, then $F\subseteq L'$.
\end{lemma}

\begin{proof}
	Let $R'$ be the final set of right vertices maintained by \cref{alg:finderasure}. Suppose for contradiction that
	\[
		S\coloneq F\setminus L'\neq \emptyset.
	\]
	Since the algorithm has stopped, every $i\in S$ satisfies $|N(i)\cap R'|<h$. Thus
	\begin{equation}
		\label{eq:viderman-edge-S-R}
		|E(S,R')|<h|S|.
	\end{equation}
	We now upper bound $|N(S)|$. Consider $v\in N(S)\setminus R'$. As $N(L')\subseteq R'$, the check $v$ has no neighbor in $L'$. 
    As $v\in N(S)\subseteq N(F)$ and $F=F(x,y)$, we have $x_{N(v)}\neq y_{N(v)}$. On the other hand, as $U(x)\subseteq R'$, the check $v$ is satisfied by $x$, implying that $x_{N(v)}$ is a codeword of $C_0$. So $x_{N(v)}$ and $y_{N(v)}$ are distinct codewords of $C_0$.
   It follows that
	\[
	|N(v)\cap F|=|N(v)\cap S|\geq d_0.
	\]
    As this holds for all $v\in N(S)\setminus R'$, we have
    \[
   |N(S)\setminus R'|\leq  \frac{|E(S,N(S)\setminus R')|}{d_0}=\frac{c|S|-|E(S,R')|}{d_0}.
    \]
	It follows that
	\begin{align*}
		|N(S)|  
		 & =|N(S)\cap R'|+|N(S)\setminus R'|                            \\
		 & \leq |E(S,R')|+\frac{c|S|-|E(S,R')|}{d_0}                    \\
		 & =\left(1-\frac{1}{d_0}\right)|E(S,R')|  +\frac{c}{d_0}|S|    \\
		 & <\left(\left(1-\frac{1}{d_0}\right)h+\frac{c}{d_0}\right)|S| \\
		 & =\delta c|S|,
	\end{align*}
	where the strict inequality follows from \eqref{eq:viderman-edge-S-R}, and the last equality follows from the definition of $h$ \eqref{eq:viderman-h}. On the other hand, since $S\subseteq F$ and $|F|\leq \alpha n$, the expansion property gives $|N(S)|\geq \delta c|S|$, a contradiction. Hence $F\subseteq L'$.
\end{proof}

We then prove the correctness of the erasure decoding algorithm.
\begin{lemma}\label{lem:erasuredecode-correct}
	Assume $\delta d_0>1$. Let $y\in T(G,C_0)$, and suppose $L'\subseteq L$ satisfies $F(x,y)\subseteq L'$ and $|L'|\leq \alpha n$. Then \cref{alg:erasuredecode}, with input $(x,L')$, outputs $y$ in $O(n)$ time.
\end{lemma}

\begin{proof}
	We maintain the invariant that $\hat x_i=y_i$ for every $i\in L\setminus L'$. This invariant holds initially because $F(x,y)\subseteq L'$.

	Suppose $L'\neq \emptyset$. Since $|L'|\leq \alpha n$, applying \cref{lem:N_let} with $S=L'$ and $t=d_0-1$ gives
	\[
		|N_{\leq d_0-1}(L')|\geq \frac{\delta d_0-1}{d_0-1}\cdot c|L'|>0.
	\]
	Thus there exists $v\in R$ such that $1\leq |N(v)\cap L'|\leq d_0-1$. In particular, Line~3 of \cref{alg:erasuredecode} can always be executed whenever $L'\neq\emptyset$, which is guaranteed by the condition in Line~2.

	By the invariant, we have $\hat x_{N(v)\setminus L'}=y_{N(v)\setminus L'}$. Since $y_{N(v)}\in C_0$, there exists a codeword of $C_0$ that agrees with $\hat x$ on the unerased coordinates $N(v)\setminus L'$. This codeword is unique. Indeed, if two distinct codewords of $C_0$ agreed with $\hat x_{N(v)\setminus L'}$ on the unerased coordinates, then their difference would be a nonzero codeword supported on  at most $|N(v)\cap L'|\le d_0-1$ coordinates, contradicting the minimum distance of $C_0$. Therefore, $\erasuredecode(\hat x_{N(v)},N(v)\cap L')$ returns $y_{N(v)}$
	and the invariant is preserved after removing $N(v)$ from $L'$.

The loop terminates after at most $|L'|$ iterations, since $|N(v)\cap L'|\ge 1$ in every iteration and all elements of $N(v)\cap L'$ are removed from $L'$. At termination, $L'=\emptyset$, and the invariant implies that $\hat x=y$.

To achieve the desired running time, maintain the values $|N(v)\cap L'|$ for all $v\in R$, together with a queue of checks $v$ satisfying $1\leq |N(v)\cap L'|\leq d_0-1$. Whenever erased coordinates are recovered, update only the counts of their neighboring checks. Since $c$, $d$, and the local erasure-decoding time of $C_0$ are constants, each edge of the graph $G$ is processed only $O(1)$ times. Therefore, the total running time is $O(n)$.
\end{proof}

We are now ready to prove \cref{thm:viderman-style} by combining the preceding lemmas.

\begin{proof}[Proof of \cref{thm:viderman-style}]
	Let $y\in T(G,C_0)$ be the transmitted codeword, and let $F=F(x,y)$. Assume $|F|\leq \gamma(\alpha n-1)$.

	By \cref{lem:finderasure-small}, the output $L'$ of \cref{alg:finderasure} satisfies $|L'|\leq \alpha n$. Since $|F|\leq \gamma(\alpha n-1)$ and $\gamma\leq 1$, we also have $|F|\leq \alpha n$, and therefore \cref{lem:finderasure-contain} gives $F\subseteq L'$. Finally, the condition $\delta>\frac{d_0}{2d_0-1}$ implies $\delta d_0>1$, so \cref{lem:erasuredecode-correct} shows that \cref{alg:erasuredecode} outputs $y$.

	It remains to verify that \cref{alg:finderasure} runs in $O(n)$ time. We first compute $U(x)$ by checking every right vertex, which takes $O(n)$ time. We then maintain the values $|N(i)\cap R'|$ for all $i\in L$, together with a queue of vertices in $L\setminus L'$ satisfying the threshold condition. Whenever a new right vertex is added to $R'$, we update the counters of its $d$ left neighbors, adding them to the queue when necessary. When a vertex from the queue is added to $L'$, it is removed from the queue.

    Every left vertex is added to the queue at most once and removed to the queue at most once. Moreover, since each right vertex is added to $R'$ at most once and $c,d$ are constants, every edge is processed only $O(1)$ times. Therefore, \cref{alg:finderasure} runs in $O(n)$ time. Combined with \cref{lem:erasuredecode-correct}, this implies that the entire decoder runs in $O(n)$ time.
\end{proof}